\documentclass[12pt]{amsart}
\usepackage[sans]{dsfont}
\usepackage{amsfonts,amssymb,amsbsy,amsmath,amsthm,enumerate}

\numberwithin{equation}{section}

\newtheorem{theorem}{Theorem}[section]
\newtheorem{proposition}[theorem]{Proposition}
\newtheorem{lemma}[theorem]{Lemma}

\newtheorem{corollary}[theorem]{Corollary}

\theoremstyle{definition}

\usepackage{amsfonts}
\usepackage{amsmath}
\usepackage{bigints}
\usepackage{bm}
\usepackage{dsfont}
\def\beq{\begin{equation}}
\def\eeq{\end{equation}}
\newcommand{\bea}{\begin{eqnarray}}
\newcommand{\eea}{\end{eqnarray}}
\newcommand{\beas}{\begin{eqnarray*}}
\newcommand{\eeas}{\end{eqnarray*}}

\newcommand{\bel}{\begin{equation} \label}
\newcommand{\ee}{\end{equation}}

\newcommand{\bethl}{\begin{theorem} \label}

\newcommand{\beprl}{\begin{proposition} \label}
\newcommand{\epr}{\end{proposition}}

\newcommand{\belel}{\begin{lemma} \label}
\newcommand{\ele}{\end{lemma}}

\newcommand{\becol}{\begin{corollary} \label}
\newcommand{\eco}{\end{corollary}}

\newcommand{\bepf}{\begin{proof} \smartqed}
\newcommand{\epf}{\qed \end{proof}}

\newcommand{\bx}{{\bf x}}
\newcommand{\bw}{{\bf w}}
\newcommand{\bt}{{\bf t}}
\newcommand{\bk}{{\bf k}}
\newcommand{\bl}{{\bf l}}

\newcommand{\bxi}{{\bm{\xi}}}
\newcommand{\btau}{{\bm{\tau}}}

\newcommand{\one}{\mathds{1}}
\newcommand{\I}{i}
\newcommand{\E}{e}
\newcommand{\D}{d}

\newcommand{\rd}{{\mathbb R}^{2}}

\newcommand{\re}{{\mathbb R}}

\newcommand{\C}{{\mathbb C}}

\newcommand{\N}{{\mathbb N}}
\newcommand{\sone}{{\mathbb S}^1}
\newcommand{\Z}{{\mathbb Z}}
\newcommand{\gB}{{\mathfrak{B}}}
\newcommand{\gD}{{\mathfrak{D}}}

\newcommand{\gh}{{\mathfrak{h}}}

\newcommand{\gC}{{\mathfrak{C}}}
\newcommand{\gM}{{\mathfrak{M}}}
\newcommand{\gS}{{\mathfrak{S}}}

\newcommand{\gV}{{\mathfrak{V}}}

\newcommand{\cC}{{\mathcal C}}
\newcommand{\cD}{{\mathcal D}}
\newcommand{\cE}{{\mathcal E}}
\newcommand{\cF}{{\mathcal F}}
\newcommand{\cG}{{\mathcal G}}

\newcommand{\cL}{{\mathcal L}}
\newcommand{\cM}{{\mathcal M}}
\newcommand{\cN}{{\mathcal N}}
\newcommand{\cO}{{\mathcal O}}
\newcommand{\cQ}{{\mathcal Q}}
\newcommand{\cR}{{\mathcal R}}
\newcommand{\cS}{{\mathcal S}}
\newcommand{\cT}{{\mathcal T}}
\newcommand{\cU}{{\mathcal U}}
\newcommand{\cV}{{\mathcal V}}
\newcommand{\cW}{{\mathcal W}}
\newcommand{\cZ}{{\mathcal Z}}

\newcommand{\pdos}{$\Psi$DOs }

\newcommand{\veps}{{\varepsilon}}

\newcommand{\opw}{{\rm Op}^{\rm w}}
\newcommand{\tqw}{{\cT}_q}
\newcommand{\opwv}{{\rm Op}^{\rm w}(\cV)}
\newcommand{\tqwv}{{\cT}_q(\cV)}
\newcommand{\opaw}{{\rm Op}^{\rm aw}}

\newcommand{\opawv}{{\rm Op}^{\rm aw}(\cV)}

\newcommand{\tvbq}{\tilde{v}_{b,q}}
\newcommand{\vbq}{v_{b,q}}
\newcommand{\wbq}{w_{b,q}}
\newcommand{\obq}{\omega_{b,q}}
\newcommand{\Homu} {H_0}

\newcommand{\gwrdn}{\Gamma_{\rm w}(\re^{2n})}
\newcommand{\gwrd}{\Gamma_{\rm w}(\rd)}
\newcommand{\gwrf}{\Gamma_{\rm w}(\re^4)}
\newcommand{\gawrdn}{\Gamma_{\rm aw}(\re^{2n})}
\newcommand{\gawrd}{\Gamma_{\rm aw}(\rd)}

\begin{document}
\title[Landau Hamiltonians with non-local potentials]{Spectral properties of Landau Hamiltonians with non-local potentials}

\author[E.~C\'ardenas]{Esteban C\'ardenas}
\author[G.~Raikov]{Georgi Raikov}
\author[I.~Tejeda]{Ignacio Tejeda}

\begin{abstract}\setlength{\parindent}{0mm}
We consider the Landau Hamiltonian $H_0$, self-adjoint in $L^2(\rd)$, whose spectrum consists of an arithmetic progression of infinitely degenerate positive  eigenvalues $\Lambda_q$, $q \in \Z_+$. We perturb $H_0$ by a non-local potential written as a bounded pseudo-differential operator  $\opw(\cV)$ with real-valued  Weyl symbol $\cV$, such that
$\opw(\cV) H_0^{-1}$ is compact. We study the spectral properties of the perturbed operator $H_{\cV} = H_0 + \opw(\cV)$. First, we construct symbols $\cV$, possessing a suitable symmetry, such that the operator $H_\cV$ admits an explicit eigenbasis in $L^2(\rd)$, and calculate the corresponding eigenvalues. Moreover, for  $\cV$ which are not supposed to have this symmetry, we study the asymptotic distribution of the eigenvalues of $H_\cV$ adjoining any given $\Lambda_q$. We find that the effective Hamiltonian in this context is the Toeplitz operator $\tqw(\cV) = p_q \opwv p_q$, where $p_q$ is the orthogonal projection onto ${\rm Ker}(H_0 - \Lambda_q I)$, and investigate its spectral asymptotics.
\end{abstract}

\maketitle


operators, eigenvalue asymptotics, logarithmic capacity

\section{Introduction}
\label{s1} \setcounter{equation}{0}
We consider the Landau Hamiltonian $H_0$, i.e. the 2D Schr\"odinger operator with  constant scalar magnetic field $b>0$, self-adjoint in $L^2(\rd)$.
We have
\bel{D132}
H_0 =
\left(-i\frac{\partial}{\partial x} +  \frac{b y}{2}\right)^2  +
\left(-i\frac{\partial}{\partial y} -  \frac{b x}{2}\right)^2, \quad (x,y) \in \rd.
\ee
 As is well known, the spectrum $\sigma(H_0)$ of the operator $H_0$ consists of eigenvalues of infinite multiplicity
    $$
    \Lambda_q : = b(2q+1), \quad q \in \Z_+ : = \{0,1,2,\ldots\},
    $$
     called {\em the Landau levels} (see \cite{Fo28, La30}).\\
     Let $\opw(\cV)$ be a bounded pseudo-differential operator ($\Psi$DO) with real-valued Weyl symbol $\cV$; then $\opwv$ is self-adjoint in $L^2(\rd)$. We assume moreover that $\opwv$  is relatively compact with respect to $H_0$, i.e. that the operator $\opwv H_0^{-1}$ is compact. Proposition \ref{np90} below contains simple sufficient conditions which guarantee the validity of these general assumptions on $\opwv$.  We will study the spectral properties of the perturbed operator
    $$
    H_\cV : = H_0 + \opwv.
    $$
    By the Weyl theorem on the invariance of the essential spectrum under relatively compact perturbations, we have
    $$
    \sigma_{\rm ess}(H_\cV) = \sigma_{\rm ess}(H_0) = \bigcup_{q=0}^\infty\left\{\Lambda_q\right\}.
    $$
    We will be interested, in particular, in the asymptotics of the discrete spectrum of $H_\cV$ near any fixed Landau level $\Lambda_q$, $q \in \Z_+$. As we will see, generically, the effective Hamiltonian which governs this asymptotic behavior is the Toeplitz-type operator
    \bel{n80a}
    \tqwv : = p_q \, \opwv \, p_q,
    \ee
     considered as an operator in ${\rm Ran}\,p_q = p_q\,L^2(\rd)$, where $p_q$ is the orthogonal projection onto ${\rm Ker}\,(H_0-\Lambda_q I)$.\\
    Let us explain briefly our motivation to study the spectral properties of the operator $H_\cV$. The so called {\em non-local potentials} defined as  appropriate integral operators play an important role in nuclear physics (see e.g. \cite{Fe58, FeWa71, Sa83}). Let us recall that any integral operator in $L^2(\re^n)$, which has a reasonable integral kernel can be represented as a
    Weyl $\Psi$DO (see e.g. \cite[Eq. (23.39)]{Sh01}). That is why, in the mathematical physics literature there is a persistent interest in Schr\"odinger operators  with non-local, in particular, pseudo-differential potentials (see e.g. \cite{Je77, BeCo77, FiPa90}).\\
    On the other hand, during the last three decades, there have been published numerous works on the spectral asymptotics for various types of  perturbations of $H_0$. For example:
     \begin{itemize}
     \item {\em Electric perturbations}, i.e. perturbations of $H_0$ by an additive real multiplier $V$ which plays the role of an electric potential, were considered in \cite{Ra90, Iv98, RaWa02, FiPu06};
         \item {\em Magnetic perturbations}, i.e. perturbations of the constant magnetic field $b$ by a variable one $\tilde{b}$, which involve a first-order differential operator, were investigated in \cite{Iv98, RoTa09};
              \item {\em Metric perturbations}, i.e. perturbations of the Euclidean metric $\left\{\delta_{jk}\right\}_{j,k=1,2}$ by a variable metric $g = \left\{g_{jk}\right\}_{j,k=1,2}$, which involve a second-order operator, were studied in \cite{Iv98, LuRa15}.
                   \end{itemize}
                   The perturbations, i.e. the quantities $V$, $\tilde{b}$ or $g$, considered in \cite{Ra90, Iv98} are of power-like decay at infinity, while those studied in \cite{RaWa02, FiPu06, RoTa09, LuRa15} are of exponential decay or compact support. Recently, several articles,
                   \cite{PuRo07, Pe09, GoKaPe16}, treated the eigenvalue asymptotics for the Landau Hamiltonian defined on the complement of a compact in $\rd$, and equipped with Dirichlet, Neumann, or Robin boundary conditions. In this geometric setting, the effective Hamiltonian which governs the eigenvalue asymptotics near the Landau level $\Lambda_q$ is an integral operator sandwiched between the projections $p_q$, quite similar to the Toeplitz operator $\tqwv$  in \eqref{n80a}.\\
    All these reasons are the source of our motivation to think of a unified approach to the spectral theory of pseudo-differential perturbations of magnetic quantum Hamiltonians. We believe that our present work could be a small but useful step in this direction. \\
    Let us discuss briefly and informally the main results of the article. As already mentioned, most of them concern the eigenvalue distribution of the discrete eigenvalues of $H_\cV$ near the Landau levels. In particular, we compare the characteristic features of the eigenvalue asymptotics for the operator $H_\cV$ with non-local potential $\opwv$, and for $H_V = H_0 + V$ with local potential $V = V(x,y) \in \re$ with $(x,y) \in \rd$; note that if $V$ is local then $\opw(V) = V$.\\
      If $\opwv$ is bounded, $\opwv H_0^{-1}$ is compact, and $\opwv \geq 0$, then the discrete eigenvalues of $H_\cV$ (resp., of $H_{-\cV}$) may accumulate at any $\Lambda_q$, $q \in \Z_+$, only from above (resp., only from below). However, in contrast to local symbols $V$, generally speaking, $\cV \geq 0$ does not imply $\opwv \geq 0$, and $\opwv \geq 0$ does not imply $\cV \geq 0$.  Put
      \bel{fin30}
      I_q^+ : = (\Lambda_q, \Lambda_{q+1}), \; q \in \Z_+, \; I_q^- : = (\Lambda_{q-1}, \Lambda_q), \; q \in \N, \; I_0^- : = (-\infty, \Lambda_0).
     \ee
      Suppose that $\opwv \geq 0$, $\sigma(H_{\pm \cV}) \cap I_q^\pm \neq \emptyset$, set $m_q^\pm : = \#\left\{\sigma(H_{\pm \cV}) \cap I_q^\pm\right\}$,  and denote by
       $\left\{\lambda_{k,q}^+(\cV)\right\}_{k=0}^{m_q^+ - 1}$ (resp., by $\left\{\lambda_{k,q}^-(\cV)\right\}_{k=0}^{m_q^- - 1}$)  the non-increasing (resp., non-decreasing) set of the eigenvalues of $H_{\cV}$ (resp., $H_{-\cV}$), lying on the interval $I_q^+$
    (resp., on $I_q^-$), $q \in \Z_+$.\\
     If $V \neq 0$ is a local potential of a definite sign, which decays at infinity and is sufficiently regular, then {\em all the Landau levels} $\Lambda_q$, $q \in \Z_+$, are accumulation points of the discrete spectrum of $H_V$, with a maximal admissible accumulation rate. More precisely, if the local potential $V \geq 0$ does not vanish identically and satisfies, say, $V \in C(\rd)$, and $\lim_{|{\bf x}| \to \infty} V({\bf x}) = 0$, then the results of \cite{RaWa02} imply $m_q^\pm = \infty$ for every $q \in \Z_+$, i.e. the discrete eigenvalues of $H_{\pm V}$ accumulate at each Landau level $\Lambda_q$, and we have
    \bel{fin1}
    \liminf_{k \to \infty} \frac{\ln{\left(\pm \left(\lambda_{k,q}^{\pm}(V) - \Lambda_q\right)\right)}}{k \, {\ln{k}}} \geq -1, \quad q \in \Z_+,
    \ee
    i.e. the  eigenvalues of $H_{\pm V}$ cannot accumulate arbitrarily fast at $\Lambda_q$, even if $V$ has a compact support.
     If we assume in addition that $\|V\|_{L^\infty(\rd)} < 2b$, the distance between any two consecutive Landau levels, then it follows from the results of \cite{KlRa09} that ${\rm Ker}\,(H_{\pm V}- \Lambda_q I) = \{0\}$, i.e. the perturbations $\pm V$ transform the infinite-dimensional subspace ${\rm Ker}\,(H_0- \Lambda_q I)$ into trivial subspaces ${\rm Ker}\,(H_{\pm V}- \Lambda_q I)$.\\
    In contrast to these properties of $H_V$ with local decaying sign-definite $V$, we construct in Proposition \ref{np51} a Schwartz-class symbol $\cV : \re^4 \to \re$ such that $\opwv \geq 0$, and the number of the eigenvalues of $H_{-\cV}$, lying on the interval $I_q^-$, $q \in \Z_+$, and counted with the multiplicities, is equal to any given $m_q \in \Z_+ \cup \{\infty\}$. If $m_q < \infty$, then ${\rm dim \, Ker}\,(H_{-\cV}- \Lambda_q I) = \infty$, i.e. the Landau level $\Lambda_q$ remains an eigenvalue of $H_{-\cV}$ of infinite multiplicity. If $m_q = \infty$, then the accumulation of $\lambda_{k,q}^-(\cV)$ at $\Lambda_q$ can be arbitrarily fast, i.e. there exists no maximal accumulation rate as in \eqref{fin1}. \\
    Further, one of the main problems dealt with in the physics literature on (non-magnetic) Schr\"odinger operators with non-local potentials, is the existence of an effective local potential $V$ which can replace the non-local one $\opwv$ in a given asymptotic regime (see e.g.
    \cite{CoArMa70, SoKi97, ChMoKiChKiSo14}). Let us describe briefly our construction which is similar to such a replacement. As already mentioned, we prove that the effective Hamiltonian which governs the eigenvalue asymptotics of $H_{\pm \cV}$ with $\opwv \geq 0$ near $\Lambda_q$ is the Toeplitz operator $\tqwv$ defined in \eqref{n80a} (see Proposition \ref{np7}). In Corollary \ref{nf1}, we show that
    $\tqwv$ is unitarily equivalent to the operator $\opw(v_{b,q})$, compact in $L^2(\re)$ whose symbol $v_{b,q} : \rd \to \re$ is a suitable integral transform of the symbol $\cV : \re^4 \to \re$ (see \eqref{j19}). Next, we make the crucial assumption that the operator $\opw(v_{b,q})$ admits an {\em anti-Wick symbol} $\tilde{v}_{b,q}$ i.e. that $v_{b,q}$ is the convolution of $\tilde{v}_{b,q}$ with a Gaussian function (see \eqref{aa1}). Then $\opw(v_{b,q})$ is unitarily equivalent to the Toeplitz operator $p_0 \omega_{b,q} p_0$ with
    \bel{fin20}
    \omega_{b,q}(x,y)  = \tilde{v}_{b,q}(-b^{1/2}y, -b^{1/2}x), \quad (x,y) \in \rd,
    \ee
    (see Corollary \ref{nf10}). Thus, $\omega_{b,q}$ could be regarded as the effective local counterpart of $v_{b,q}$ in the asymptotic analysis of the eigenvalue distribution near $\Lambda_q$, $q \in \Z_+$, for the operator $H_\cV$. We assume further that there exist $r \in \Z_+$ and $\zeta_{b,q,r} : \rd \to [0,\infty)$ which does not vanish identically and has a compact support or decays exponentially at infinity, such that
    \bel{d16}
    \omega_{b,q} = \cD_{b,r} \, \zeta_{b,q,r},
    \ee
    where $\cD_{b,0}$ is the identity, and ${\cD}_{b,r}$ with $r \geq 1$  is a partial differential operator of order $2r$ (see \eqref{d15}). The passage from $\omega_{b,q}$ to $\zeta_{b,q,r}$ is motivated by the fact that $\zeta_{b,q,r}$ might be non-negative even if $\omega_{b,q}$ does not have a definite sign. Then it  follows from \cite[Section 9]{BrPuRa04}, that $p_0 \omega_{b,q} p_0$ is unitarily equivalent to the Toeplitz operator
    $p_r \zeta_{b,q,r} p_r$. Using this fact, we obtain several asymptotic terms of $\ln{(\pm(\lambda_{k,q}^\pm(\cV) - \Lambda_q))}$ as $k \to \infty$ for compactly supported
    $\zeta_{b,q,r}$ (see Theorem \ref{nth1}), and for of exponentially decaying $\zeta_{b,q,r}$ (see Theorem \ref{nth2}).  Due to the fast decay of  $\zeta_{b,q,r}$, the standard pseudo-differential techniques are not applicable and that is why we use and develop the methods of \cite{RaWa02}, \cite{FiPu06}, and \cite{LuRa15}.
    Finally, we drop our assumption that $\opwv$ has a definite sign but assume that $v_{b,q}$ has a power-like decay at infinity, and in Theorem \ref{oth1} we obtain the main asymptotic term of the local eigenvalue counting function as the energy approaches the Landau level $\Lambda_q$, $q \in \Z_+$. Here, many traditional $\Psi$DO techniques are applicable, in particular, we use and extend the methods developed in \cite{Ra90} and \cite{DaRo87}.\\
    The article is organized as follows. In Section \ref{s2} we summarize the necessary facts from the general theory of $\Psi$DOs with Weyl and anti-Wick symbols. Section \ref{s3} contains the description of several unitary operators which map $H_\cV$ to operators which are more accessible and easier to investigate. In particular,  we show that  $\tqwv$, $q \in \Z_+$, is unitarily equivalent to  $\opw(v_{b,q})$. In Section \ref{s5} we deal with Weyl and anti-Wick $\Psi$DOs with radial symbols, and obtain explicit formulas for their eigenvalues and eigenfunctions. Some of these results are known, others, to our best knowledge, are new and could be of independent interest.  As a corollary, we construct a family of symbols $\cV : \re^4 \to \re$, possessing a suitable symmetry such that the operator $H_\cV$ has explicit eigenvalues and orthonormal basis of eigenfunctions.
    In Section \ref{s4} we consider the eigenvalue distribution near the Landau level $\Lambda_q$, $q \in \Z_+$, for the operator $H_\cV$. First, in Proposition \ref{np51} we construct our explicit example of a symbol $\cV$ in the Schwartz class $\cS(\re^4)$ such that $\opwv \geq 0$, which shows that the asymptotic behavior of the discrete spectrum of the operator $H_{-\cV}$ near a given $\Lambda_q$ could be arbitrarily fast in contrast to the case of a local potential $V$. Next, we examine the eigenvalue asymptotics for the operators $H_{\pm \cV}$ with $\opwv \geq 0$, assuming that the operator $\opw(v_{b,q})$ admits an anti-Wick symbol $\tilde{v}_{b,q}$ related to $\zeta_{b,q,r}$ through \eqref{fin20} and \eqref{d16}, and that $\zeta_{b,q,r}$ decays exponentially at infinity or has a compact support.
    Finally, Theorem \ref{op1} contains our result on the eigenvalue asymptotics for $H_\cV$ near the Landau level $\Lambda_q$, $q \in \Z_+$, in the case where $\opwv$ is not supposed to have a definite sign but $v_{b,q}$ has a power-like decay at infinity.

 \section{Weyl and anti-Wick $\Psi$DOs}
\label{s2}
     In this section we recall briefly some basic facts from the theory of $\Psi$DOs with Weyl and anti-Wick symbols, assuming that the dimension $n \geq 1$.
      We will use the following notations. Let $X$ be a separable Hilbert space with scalar product $\langle \cdot, \cdot\rangle_X$, linear with respect to the first factor, and norm $\|\cdot\|_X$. By $\gB(X)$ (resp., by $\gS_\infty(X))$ we will denote the space of linear bounded (resp., compact) operators in $X$, and by $\gS_p(X)$, $p \in [1,\infty)$, the $p$th Schatten-von Neumann space of operators $T \in \gS_\infty(X))$ for which the norm
      $\|T\|_p : = \left({\rm Tr}\,(T^* T)^{p/2}\right)^{1/p}$
      is finite. In particular, $\gS_1(X)$ is the trace class, and $\gS_2(X)$ is the Hilbert-Schmidt class\\
      Let $\cS(\re^n)$ be the Schwartz class over $\re^n$, and $\cS'(\re^n)$ be its dual class. For $\cF \in \cS(\re^{2n})$ we define the $\Psi$DO $\opw(\cF)$  with Weyl symbol $\cF$ as the operator with integral kernel
      \bel{1000}
      K(\bx, \bx') = (2\pi)^{-n}\,\int_{\re^{n}}\cF\left(\frac{{\bf x} + {\bf x}'}{2}, \bxi\right)
\E^{\I({\bf x}-{\bf x}') \cdot \bxi}\, \D\bxi, \quad \bx, \bx' \in \re^n.
    \ee
Let $u,v \in \cS(\re^n)$. Define {\em the Wigner transform} $W(u,v)$ of the pair $(u,v)$ by
    $$
    (W(u,v))(\bx,\bxi) : = (2\pi)^{-n} \int_{\re^n} \E^{\I\bx'\cdot\bxi}u(\bx-\bx'/2) \, \overline{v(\bx+\bx'/2)}\,\D\bx', \quad (\bx, \bxi) \in \re^{2n}.
    $$
    Then $W(u,v) \in \cS(\re^{2n})$ and we have $W(v,u) = \overline{W(u,v)}$. Moreover, the Wigner transform extends to $u,v \in L^2(\re^n)$ in which case
    $$
    \|W(u,v)\|_{L^2(\re^{2n})}^2 = (2\pi)^{-n} \, \|u\|_{L^2(\re^n)}^2 \, \|v\|_{L^2(\re^n)}^2.
    $$
    By \cite[Eq. (23.39)]{Sh01}, the function $(2\pi)^n W(u,v)$ coincides with the Weyl symbol of the operator with integral kernel $u(\bx) \overline{v(\bx')}$, $\bx, \bx' \in \re^n$.
    Note that if $\cF \in \cS(\re^{2n})$ and  $u,v \in \cS(\re^n)$, then
\bel{j10}
\langle \opw(\cF) u, v \rangle_{L^2(\re^n)}  =  \langle \cF , W(v,u) \rangle_{L^2(\re^{2n})}.
\ee
    Therefore, if
$\cF \in \cS'(\re^{2n})$, then \eqref{j10} defines a linear continuous mapping $\opw(\cF) : \cS(\re^n) \to \cS'(\re^n)$.\\
Let us now introduce the Fourier transform
$$
(\Phi u)(\bxi) = \hat{u}(\bxi) : = (2\pi)^{-N/2} \int_{\re^N} e^{-i \bxi \cdot \bx} f(\bx)\,d\bx, \quad \bxi \in \re^N,
$$
for $u \in \cS(\re^N)$, $N \geq 1$, and then extend it to $\cS'(\re^N)$. In particular, $\Phi$ extends to a unitary operator in $L^2(\re^N)$. \\
If $\cF \in \cS(\re^{2n})$, then the integral kernel of the operator $\opw(\cF)$ can be written not only as in \eqref{1000} but also as
    \bel{1001}
K(\bx, \bx') = (2\pi)^{-n}\,\int_{\re^{n}}\widehat{\cF}\left(\bxi, {\bf x}' - {\bf x}\right)
\E^{\I(\frac{{\bf x} + {\bf x}'}{2}) \cdot \bxi}\, \D\bxi, \quad \bx, \bx' \in \re^n.
    \ee

    Let ${\Gamma_{\rm w}(\re^{2n})}$, $n\geq 1$,
denote the set of functions $\cF: \re^{2n} \to {\mathbb C}$ such that
    $$
    \|\cF\|_{\Gamma_{\rm w}(\re^{2n})} : = \sup_{\{\alpha , \beta \in {\mathbb
Z}_+^n \; | \; |\alpha|, |\beta| \leq  [\frac{n}{2}] + 1\}}
\sup_{(\bx,\bxi) \in \re^{2n}} |D_\bx^{\alpha}
D_{\bxi}^{\beta} \cF(\bx,\bxi)| < \infty.
    $$
    Note that $\Gamma_{\rm w}(\re^{2n}) \subset \cS'(\re^{2n})$.

 \beprl{Dprvbp1}
{\rm  (\cite{CaVa72}, \cite{Co75}, \cite[Corollary 2.5 (i)]{Bo99})} Let $\cF \in \gwrdn$. Then $\opw(\cF)$ extends to an operator bounded in $L^2(\re^n)$. Moreover, there exists a constant $c_0$ independent of $\cF$, such that
$$
     \|\opw(\cF)\| \leq c_0\|\cF\|_{\Gamma_{\rm w}(\re^{2n})}.
    $$
    \epr
    {\em Remark}: We will consider Weyl $\Psi$DOs $\opw(\cF)$ acting in $L^2(\re^n)$, under the generic assumption $\cF \in \gwrdn$;  then, by Proposition \ref{Dprvbp1}, we have $\opw(\cF) \in \gB(L^2(\re^n))$. However, many assertions in the sequel remain valid under more general assumptions about  $\cF$.\\

 Further, for $m \in \re$ and $\varrho \in (0, 1]$, introduce the H\"ormander--Shubin class
$$
\cS_\varrho^{m}(\re^N) : = \left\{u \in C^\infty(\re^N) \, | \, \sup_{\bx \in \re^N} \langle \bx \rangle^{-m + \varrho |\alpha|}|D^\alpha\,u(\bx)| < \infty , \; \alpha \in \Z_+^N\right\}. $$
\beprl{Dprvbp1b} {\rm \cite[Problem 24.9]{Sh01}}
Let $\cF \in \cS_\varrho^{0}(\re^{2n})$ with $\varrho \in (0, 1]$. Assume that
    $$
   \lim_{|{\bf w}| \to \infty} {\mathcal F}({\bf w}) = 0,
   $$
uniformly with respect to $\frac{{\bf w}}{|{\bf w}|} \in {\mathbb S}^{2n-1}$. Then $\opw(\cF) \in \gS_\infty(L^2(\re^n))$.
\epr
The Parseval theorem and \eqref{1000} imply the following
\beprl{Dprvbp1c}
Let $\cF \in L^2(\re^{2n})$. Then $\opw(\cF)$ extends to a Hilbert-Schmidt operator in $L^2(\re^n)$, and
    \bel{d10}
     \|\opw(\cF)\|_2^2 = (2\pi)^{-n} \|\cF\|^2_{L^2(\re^{2n})}.
    \ee
    \epr
    Next, we describe {\em the  metaplectic unitary equivalence}
of Weyl \pdos   whose symbols are mapped into each other by a linear
symplectic transformation.
    \beprl{Dp42}
{\em \cite[Chapter 7, Theorem A.2]{DiSj99}}
Let $\kappa: \re^{2n} \rightarrow \re^{2n}$, $n\geq 1$,  be a
linear symplectic transformation. Assume that $\cF \in \gwrdn$. Then there exists a unitary operator $\gM_\kappa:
L^2(\re^n) \rightarrow L^2(\re^n)$ such that
    \bel{D31}
\opw(\cF \circ \kappa) = \gM_\kappa^* \opw(\cF) \gM_\kappa.
    \ee
    \epr
{\em Remarks}: (i) Proposition \ref{Dp42} remains valid for a considerably wider class of symbols including the linear and the quadratic ones.\\
(ii) The operator $\gM_\kappa$ is called {\em the metaplectic operator} generated by the linear symplectomorphism $\kappa$.\\

Further, we discuss {\em the anti-Wick} $\Psi$DOs. Let at first $\cF \in \gwrdn$. Set
    \bel{ja6}
\cG_n({\bf w}) : = \pi^{-n} e^{-|{\bf w}|^2}, \quad  {\bf w} \in \re^{2n},
    \ee
and define the $\Psi$DO
    \bel{aa1}
\opaw(\cF)  : =  \opw(\cF * \cG_n).
    \ee
Then we will say that $\opaw(\cF)$ is a $\Psi$DO with anti-Wick symbol $\cF$.  If $\cF \in \cS(\re^{2n})$ and $u, v \in \cS(\re^{n})$, then, similarly to \eqref{j10}, we have
\bel{d13}
\langle \opaw(\cF) u, v \rangle_{L^2(\re^n)}  =  \langle \cF , \cG_n * W(v,u) \rangle_{L^2(\re^{2n})}
\ee
where
$ \cG_n * W(u,v) \in \cS(\re^{2n})$
is the {\em Husimi transform} of $(u,v)$.
    Therefore, if
$\cF \in \cS'(\re^{2n})$, then \eqref{d13} defines, similarly to \eqref{j10}, a linear continuous mapping $\opaw(\cF) : \cS(\re^n) \to \cS'(\re^n)$.\\
Since the convolution with the Gaussian function $\cG_n$ may improve the regularity and the decay rate of the symbol $\cF$, the definition of the anti-Wick $\Psi$DOs can be  extended to a  class of symbols, considerably larger than $\gwrdn$. In particular, we have the following

\beprl{np3} {\rm \cite[Lemma 2.5]{PuRaVBl13}}
{\rm (i)} Let $\cF \in L^\infty(\re^{2n})$. Then $\opaw(\cF) \in {\gB}(L^2(\re^n))$, and
$$
\|\opaw(\cF)\| \leq \|\cF\|_{L^\infty(\re^{2n})}.
$$
{\rm (ii)} Let $\cF \in L^p(\re^{2n})$, $p \in [1,\infty)$. Then $\opaw(\cF) \in {\gS}_p(L^2(\re^n))$, and
$$
\|\opaw(\cF)\|^p_p \leq (2\pi)^{-n} \|\cF\|^p_{L^p(\re^{2n})}.
$$
\epr
Set
$$
\gawrdn : = L^1(\re^{2n}) + L^\infty(\re^{2n}).$$
Note that if $\cF \in \gawrdn$, then $\cF * \cG_n \in \gwrdn$. Our generic assumption concerning anti-Wick $\Psi$DOs $\opaw(\cF)$ will be $\cF \in \gawrdn$.
As in the case of Weyl $\Psi$DOs, many assertions in the sequel hold true under wider assumptions.\\
Let us give an alternative definition of the anti-Wick $\Psi$DO $\opaw(\cF)$ with $\cF \in \gawrdn$. For $({\bf x}, \bxi) \in \re^{2n}$ set
$$
\phi_{{\bf x}, \bxi}({\bf y}) : = \pi^{-n/4} e^{i\bxi \cdot {\bf y}} e^{-\frac{|{\bf x} - {\bf y}|^2}{2}}, \quad {\bf y} \in \re^n,
$$
and introduce the rank-one orthogonal projection
$$
P_{{\bf x}, \bxi} : = \langle \cdot, \phi_{{\bf x}, \bxi}\rangle_{L^2(\re^n)} \, \phi_{{\bf x}, \bxi}.
$$
Then we have
    \bel{n4}
\opaw(\cF) = (2\pi)^{-n} \int_{\re^{2n}} \cF({\bf x}, \bxi) P_{{\bf x}, \bxi} \, d{\bf x}\,d\bxi,
    \ee
    where the integral is understood in the weak sense. Identity \eqref{n4} implies the monotonicity of $\opaw(\cF)$ with respect to the symbol $\cF$. Namely, we have the following important
    \beprl{np5} Assume that $\cF \in \gawrdn$, and $\cF(\bx, \bxi) \geq 0$ for almost every $(\bx, \bxi) \in \re^{2n}$. Then
    $\opaw(\cF) \geq 0.$
     \epr

     {\em Remark}: Not every Weyl $\Psi$DO $\opw(\cF)$ admits an anti-Wick symbol $\widetilde{\cF}\in \cS'(\re^{2n})$. If $\cF$ is a given Weyl symbol, then in order to find the corresponding anti-Wick symbol $\widetilde{\cF}$ we have to solve the equation
     \bel{n81}
     \cF = \widetilde{\cF} * \cG_n,
     \ee
     i.e.  to invert the so called Weierstrass transform, or, which is equivalent, to solve the inverse heat equation (see \cite[Remark 24.2]{Sh01}). For example, if $\cF \in C_0^\infty(\re^{2n})$ and $\cF \neq 0$, then there exists no $\widetilde{\cF} \in \cS'(\re^{2n})$ such that \eqref{n81} holds true. On the other hand,
if the Fourier transform $\widehat{\cF}$ of $\cF$ is in $C_0^\infty(\re^{2n})$, then $\opw(\cF)$ admits an anti-Wick symbol $\widetilde{\cF} \in \cS(\re^{2n})$ given by
    $$
    \widetilde{\cF}(\bw) = (2\pi)^{-n} \int_{\re^{2n}} e^{i {\bf u}\cdot\bw}e^{|{\bf u}|^2/4} \widehat{\cF}({\bf u})\,d{\bf u}, \quad \bw \in \re^{2n}.
    $$

\section{Unitary equivalences of the operators $H_\cV$}
\label{s3}
In this section we establish  unitary equivalences for the Landau Hamiltonian $H_0$ and its perturbation $\opwv$.
First, we describe a suitable spectral representation of $H_0$. \\
Let $\varphi(x,y) : = \frac{b(x^2 + y^2)}{4}$, so that $\Delta \varphi = b$.
Introduce the magnetic creation operator
\bel{D132c}
a^* = -2i e^{\varphi} \frac{\partial}{\partial z} e^{-\varphi}, \quad z = x+iy,
\ee
and the magnetic annihilation operator
\bel{D132a}
a = -2i \E^{-\varphi} \frac{\partial}{\partial \bar{z}} \E^{\varphi}, \quad \bar{z} = x-iy.
\ee
The operators $a$ and $a^*$ are closed on $\gD(a) = \gD(a^*) = \gD(H_0^{1/2})$, and are mutually adjoint in $L^2(\rd)$. Moreover,
	\bel{Da31}
[a, a^*] = 2b \; I,
	\ee
 and
\bel{D132b}
\Homu = a^* a + b  I = a a^* - b I.
\ee
 Therefore,
 \bel{Dj1}
    {\rm Ker}\,(\Homu - \Lambda_q I) = (a^*)^q\, {\rm Ker}\,a, \quad q \in \Z_+,
    \ee
    and, by \eqref{D132a}, we have
    $$
    {\rm Ker}\;a = \left\{u \in L^2(\rd) \, | \, u = \E^{-\varphi}\,g, \quad \frac{\partial g}{\partial \bar{z}} = 0\right\}.
    $$
    Up to the unitary mapping $g \mapsto e^{-\varphi} g$, ${\rm Ker}\;a$ coincides with {\em Fock-Segal-Bargmann space} of holomorphic functions (see e.g. \cite[Section 3.2]{Ha00}).\\
Next, we recall that $\Homu$ is unitarily equivalent under an appropriate metaplectic mapping to the operator $(b \gh) \otimes I_y$, where
	\bel{o18}
\gh:= - \frac{d^2}{dx^2} + x^2
	\ee
is the harmonic oscillator, self-adjoint in $L^2(\re_x)$, and essentially self-adjoint on $C_0^\infty(\re)$, while $I_y$ is the identity in $L^2(\re_y)$. Let us recall the spectral properties of $\gh$.
We have
$$
\gh= \alpha^* \alpha + I = \alpha \alpha^* - I,
$$
where
$$
\alpha:= -i\frac{d}{dx} - ix, \quad \alpha^*: = -i\frac{d}{dx} + ix,
$$
are the standard annihilation and creation operators which are closed on $\gD(\alpha) = \gD(\alpha^*) = \gD(\gh^{1/2})$, and are mutually adjoint in $L^2(\re)$.
Moreover, they satisfy the commutation relation
$$
[\alpha, \alpha^*] = 2 I.
$$
 Therefore,
$$
\sigma(\gh) = \bigcup_{q \in \Z_+} \{2q+1\},
$$
$$
{\rm Ker}\, (\gh-(2q+1)I) = (\alpha^*)^q {\rm Ker}\,\alpha, \quad q \in \Z_+ : = \{0,1,2,\ldots\}.
$$
Since
$$
{\rm Ker}\,\alpha = \left\{u \in L^2(\re)\,|\,u(x) = c\E^{-x^2/2}, x \in \re, \quad c \in \C\right\},
$$
we have
$$
{\rm dim \, Ker}\, (h-(2q+1)I) = 1, \quad q \in \Z_+.
$$
Denote by $\pi_q$ the orthogonal projection onto ${\rm Ker}\,(\gh - (2q+1)I)$, $q \in \Z_+$.
Set
$$
\tilde{\psi}_q(x): = \left(-\frac{d}{dx}+x\right)^q \E^{-x^2/2} = (-i)^q \, (\alpha^*)^q \, \E^{-x^2/2}, \quad x \in \re, \quad q \in \Z_+.
$$
Then the functions $\tilde{\psi}_q$ satisfy $$h \tilde{\psi}_q = (2q+1) \tilde{\psi}_q, \quad q \in \Z_+,$$ and form an orthogonal eigenbasis in $L^2(\re)$.
A simple calculation shows that
 \bel{Dradi}
\psi_q: = \tilde{\psi}_q/\|\tilde{\psi}_q\| =
\frac{{\rm H}_{q}(x) \E^{-x^2/2}}{(\sqrt{\pi}2^{q} q!)^{1/2}}, \quad
x \in \re,
\ee
where
$$
{\rm H}_q(x): = e^{x^2/2} \left(-\frac{d}{dx}+x\right)^q e^{-x^2/2} = (-1)^q \E^{x^2} \frac{d^q}{dx^q} \E^{-x^2}, \quad
x \in \re,
$$
  is the Hermite
polynomial of degree $q$ (see e.g. \cite{GrRy65}). Thus, the functions $\psi_q$, $q \in \Z_+$, form an orthonormal basis in $L^2(\re)$.
 Introduce the Wigner functions
    \bel{j15}
    \Psi_{j,k} : = W(\psi_j, \psi_k), \quad j,k \in \Z_+.
    \ee
    If $j=k$, we will write
    \bel{j15a}
    \Psi_k = \Psi_{k,k}, \quad k \in \Z_+.
    \ee
Lemma \ref{jp3} below contains explicit expressions for $\Psi_{k,\ell}$, $k, \ell \in \Z_+$. In order to formulate it, we introduce
the (generalized) Laguerre polynomials
\bel{ja5}
    {\rm L}_q^{(\nu)}(\xi): = \frac{\xi^{-\nu} \E^{\xi}}{q!} \frac{d^q}{d\xi^q}
\left(\xi^{q+\nu}  \E^{-\xi}\right), \quad \xi>0, \quad \nu \in \re, \quad q \in \Z_+.
\ee
As usual, we will write
\bel{Dj3}
   {\rm L}_q(\xi) : = {\rm L}^{(0)}_q(\xi) = \frac{e^{\xi}}{q!} \frac{d^q}{d\xi^q}
\left(\xi^q e^{-\xi}\right) = \sum_{\ell = 0}^q {q \choose \ell} \frac{(-\xi)^\ell}{\ell!}, \quad \xi \in \re.
    \ee
     \belel{jp3}
     Let $k,\ell \in \Z_+$. Then for $(x,\xi) \in \rd$ we have
     $$
     \Psi_{k,\ell}(x,\xi) =
     $$
     \bel{m4}
     \left\{
     \begin{array} {l}
     \frac{1}{\pi} (-1)^\ell 2^{\frac{k-\ell}{2}} \left(\frac{\ell !}{k !}\right)^{1/2} (x + \I \xi)^{k-\ell} {\rm L}_\ell^{(k-\ell)}(2(x^2 + \xi^2))\E^{-(x^2+\xi^2)}, \quad k \geq \ell,\\ [3mm]
     \frac{1}{\pi} (-1)^k 2^{\frac{\ell-k}{2}} \left(\frac{k !}{\ell !}\right)^{1/2} (x - \I \xi)^{\ell-k} {\rm L}_k^{(\ell-k)}(2(x^2 + \xi^2))\E^{-(x^2+\xi^2)}, \quad k \leq \ell,
     \end{array}
     \right.
     \ee
     In particular,
     \bel{n55}
     \Psi_{k,\ell}(r\cos{\theta},r\sin{\theta}) = \E^{\I(k-\ell)\theta} \Phi_{k,\ell}(r), \quad k,\ell \in \Z_+, \quad \theta \in [0,2\pi), \quad r \in [0,\infty),
     \ee
     where $\Phi_{k,\ell}(r)$ is a symmetric real valued matrix. Moreover,
     \bel{n54}
     \Psi_{k}(x,\xi) = \frac{1}{\pi} (-1)^k   {\rm L}_k(2(x^2 + \xi^2))e^{-(x^2+\xi^2)}, \quad k \in \Z_+, \quad (x,\xi) \in \rd.
     \ee
     \ele
     \begin{proof}
     An elementary calculation taking into account the parity of the Hermite polynomials easily yields
     \bel{m1}
     \Psi_{k,\ell}(x,\xi) = \frac{(-1)^k}{(2\pi) \sqrt{\pi} (k! \ell !)^{1/2}2^{\frac{k+\ell}{2}}} \E^{-(x^2+\xi^2)} \, \int_\re \E^{-\left(\frac{y}{2}+\I\xi\right)^2} {\rm H}_k\left(\frac{y}{2}-x\right) {\rm H}_\ell\left(\frac{y}{2}+x\right) \D y.
     \ee
     Changing the variable $\frac{y}{2} + \I \xi = t$, and applying a standard complex-analysis argument showing that we can replace the interval of integration $\re + \I \xi$ by $\re$, we find that
     \bel{m2}
     \int_\re \E^{-\left(\frac{y}{2}+\I\xi\right)^2} {\rm H}_k\left(\frac{y}{2}+x\right) {\rm H}_\ell\left(\frac{y}{2}-x\right) \D y =
     2 \int_\re \E^{-t^2} {\rm H}_k\left(t-x-\I \xi\right) {\rm H}_\ell\left(t+x-\I\xi\right) \D t.
     \ee
     By \cite[Eq. (7.377)]{GrRy65},
     $$
     \int_\re \E^{-t^2} {\rm H}_k\left(t-x-\I \xi\right) {\rm H}_\ell\left(t+x-\I\xi\right) \D t =
     $$
     \bel{m3}
     \left\{
     \begin{array} {l}
     2^k \sqrt{\pi} \ell ! (-x - \I \xi)^{k-\ell} {\rm L}_\ell^{(k-\ell)}(2(x^2 + \xi^2)), \quad k \geq \ell,\\ [3mm]
     2^\ell  \sqrt{\pi} k! (x - \I \xi)^{\ell-k} {\rm L}_k^{(\ell-k)}(2(x^2 + \xi^2)), \quad k \leq \ell.
     \end{array}
     \right.
     \ee
     Putting together \eqref{m1}, \eqref{m2}, and \eqref{m3}, we obtain \eqref{m4}.
     \end{proof} {\em Remark}: By \eqref{n54} with $q=0$,  we have
     \bel{ja1}
     \Psi_0(x,\xi) = \frac{1}{\pi}  e^{-(x^2+\xi^2)} = \cG_1(x,\xi), \quad (x,\xi) \in \rd,
     \ee
     where $\cG_1$ is the Gaussian defined in \eqref{ja6}.\\

For ${\bf x} = (x,y) \in \rd, \;  \bxi = (\xi,\eta)
\in \rd$, set
    \bel{Dsof22}
     {\kappa}_{b}({\bf x}, \bxi): =
\left(\frac{1}{\sqrt{b}} (x-\eta), \frac{1}{\sqrt{b}} (\xi-y),
\frac{\sqrt{b}}{2}(\xi+y), -\frac{\sqrt{b}}{2}(\eta+x)\right).
    \ee
Evidently, the mapping ${\kappa}_b$ is linear and
symplectic. Introduce the Weyl symbol
\bel{Dgr1}
{\mathcal H}_0(\bx,\bxi) = \left(\xi+\frac{1}{2} b y\right)^2+\left(\eta -\frac{1}{2}
bx\right)^2,  \quad \bx  = (x,y) \in
\rd, \quad \bxi  = (\xi,\eta) \in \rd,
    \ee
    of the operator $\Homu$ defined  in \eqref{D132}. Then we have
    \bel{Da12a}
    ({\mathcal H}_0 \circ \kappa_b)(\bx,\bxi) = b(x^2 + \xi^2), \quad (\bx,\bxi) \in T^*\re^d.
    \ee
    Note that the function on the r.h.s. of \eqref{Da12a} coincides with the Weyl symbol of the operator $(b \gh)\otimes I_y$.
Next, define the unitary operator
$\cU_b: L^2(\rd_{x,y}) \to L^2(\rd_{x,y})$ by
\bel{D133}
(\cU_b u)(x,y): = \frac{\sqrt{b}}{2\pi}
\int_{\rd} \E^{\I\phi_b(x,y;x',y')} u(x',y') dx'dy'
\ee
where
$$
\phi_b(x,y;x',y'): = b \frac{xy}{2} +
b^{1/2}(xy' - yx') - x'y'.
$$
 Writing ${\kappa}_{b}$ as a product of elementary linear
symplectic transformations (see e.g. \cite[Lemma 18.5.8]{HoIII94}), and composing the corresponding elementary metaplectic operators, we
easily  check that
 ${\mathcal U}_b$ is a metaplectic operator generated by
 the symplectic mapping ${\kappa}_b$ in \eqref{Dsof22}.

 \beprl{jp1}
 We have
 \bel{D135}
\cU_b^* \Homu \cU_b = (b \gh ) \otimes I_y,
\ee
    \bel{D134}
\cU_b^* a  \cU_{b} = (\sqrt{b} \alpha) \otimes I_y, \quad \cU_b^* a^*  \cU_{b} = (\sqrt{b} \alpha^*) \otimes I_y.
    \ee
%and for each $q \in \Z_+$,
%\bel{j17}
%\cU_b^* p_q \cU_b = (\pi_q ) \otimes I_y,
%\ee
%$p_q$ being the orthogonal projection onto ${\rm Ker}\,(H_0 - \Lambda_q)$. \\
Moreover, if $\cV \in \gwrf$, then
\bel{j4}
\cU_b^* \opw(\cV) \cU_b = \opw(\cV_b)
\ee
where
    \bel{j18}
    \cV_b : = \cV \circ \kappa_b.
    \ee
    \epr
    \begin{proof}
    Relation \eqref{D135} (resp., \eqref{j4}) follows from Proposition \ref{Dp42} and the remark after it, and \eqref{Da12a} (resp., \eqref{j18}).
    Similarly, relations \eqref{D134} follow from Proposition \ref{Dp42} and the fact that the Weyl symbol of the operator $a$ (resp., $a^*$) is mapped under the symplectic transformation $\kappa_b$ into the symbol of the operator $(\sqrt{b} \alpha) \otimes I_y$ (resp., $(\sqrt{b} \alpha^*) \otimes I_y$).
    \end{proof}

    \beprl{np90}
    Assume that  $\cV \in \cS^0_\varrho(\re^4)$ with $\varrho \in (0,1]$,  and
    \bel{n100}
    \lim_{x^2+y^2+\xi^2+\eta^2 \to \infty} \cV_b(x,y,\xi,\eta)\,(x^2 + \xi^2)^{-1} = 0,
    \ee
   uniformly with respect to the variables on ${\mathbb S}^3$.  Then, $\opwv$ is bounded, and $\opwv H_0^{-1}$ is compact in $L^2(\rd)$.
    \epr
    \begin{proof}
    Since $\cS^0_\varrho(\re^4) \subset \gwrf$, the boundedness of $\opwv$ follows from Proposition \ref{Dprvbp1}. By Proposition \ref{jp1}, we have
    \bel{n101}
    \cU_b^* \opw(\cV) H_0^{-1} \cU_b = \opw(\cV_b) ((b\gh)^{-1} \otimes I_y).
    \ee
    By the pseudo-differential calculus, we easily find that the Weyl symbol of the operator $\opw(\cV_b) ((b\gh)^{-1} \otimes I_y)$ is in the class $\cS^0_\varrho(\re^4)$, while \eqref{n100} guarantees that this symbol decays at infinity. Then Proposition \ref{Dprvbp1b} implies
    that $\opw(\cV_b) ((b\gh)^{-1} \otimes I_y) \in \gS_\infty(L^2(\rd))$, and by \eqref{n101} we find that the operator $\opwv H_0^{-1}$ is compact as well.
    \end{proof}

    Our next goal is to establish the unitary equivalence between $\opwv$ and an operator $\cM : \ell^2(\Z_+^2) \to \ell^2(\Z_+^2)$. Similarly, we will establish the unitary equivalence between the Toeplitz operator $\tqwv$ with fixed $q \in \Z_+$, defined in \eqref{n80a}, and an operator $\cM_q : \ell^2(\Z_+) \to \ell^2(\Z_+)$.  To this end, we need the canonical basis $\left\{\varphi_{k,q}\right\}_{k \in \Z_+}$ of ${\rm Ran}\, p_q$, $q \in \Z_+$.
    Let at first $q=0$. Then the functions
    $$
    \widetilde{\varphi}_{k,0}({\bf x}) = z^k \E^{-b|{\bf x}|^2/4}, \quad {\bf x} = (x,y) \in \rd, \quad z = x+iy, \quad k \in \Z_+,
    $$
    form a natural orthogonal basis of ${\rm Ker}\;a = {\rm Ran}\,p_0$ (see e.g. \cite[Sections 3.1 -- 3.2]{Ha00}). Normalizing, we obtain the following orthonormal basis of ${\rm Ran}\,p_0$:
    \bel{j25}
    {\varphi}_{k,0}({\bf x}) : = \frac{ \widetilde{\varphi}_{k,0}({\bf x})}{\| \widetilde{\varphi}_{k,0}\|_{L^2(\rd)}} =
    \sqrt{\frac{b}{2\pi}} \sqrt{\frac{1}{k!}} \left(\sqrt{\frac{b}{2}} \, z\right)^k\, \E^{-b|{\bf x}|^2/4}, \quad \bx \in \rd, \quad k \in \Z_+.
    \ee
    Let now $q \geq 1$. Set
    \bel{j26}
   \widetilde{{\varphi}}_{k,q} = (a^*)^q \, \varphi_{k,0}, \quad k \in \Z_+.
    \ee
    The commutation relation \eqref{Da31} easily implies
    $$
    \langle \widetilde{{\varphi}}_{k,q}, \widetilde{{\varphi}}_{\ell,q}\rangle_{L^2(\rd)} = (2b)^q q! \delta_{k\ell}, \quad k,\ell \in \Z_+.
    $$
    Therefore, the functions
    \bel{j27}
    {\varphi}_{k,q} : = \frac{ \widetilde{\varphi}_{k,q}}{\| \widetilde{\varphi}_{k,q}\|_{L^2(\rd)}} = \frac{ \widetilde{\varphi}_{k,q}}{\sqrt{(2b)^q q!}}, \quad k \in \Z_+,
    \ee
    form an orthonormal basis of ${\rm Ran}\,p_q$, $q \in \N$. \\

   {\em Remark}: The functions $ {\varphi}_{k,q}$ admit an explicit expression,
    namely
    $$
    \varphi_{k,q}({\bf x}) =
    $$
    $$
    \frac{1}{i^q} \sqrt{\frac{b}{2\pi}} \sqrt{\frac{q!}{k!}} \left(\sqrt{\frac{b}{2}} \, z\right)^{k-q}\,{\rm L}_q^{(k-q)}\left(\frac{b |\bx|^2}{2}\right) \E^{-b|{\bf x}|^2/4}, \quad {\bf x} \in \rd, \quad k, q \in \Z_+,
    $$
the Laguerre polynomials being defined in \eqref{ja5}.\\

    Let $\cV \in \gwrf$. Set
    \bel{j23}
    m_{k,\ell;q,r}(\cV) : = \langle \opw(\cV) \varphi_{\ell,r} , \varphi_{k,q}\rangle_{L^2(\rd)}, \quad m_{k,\ell;q}(\cV) : = m_{k,\ell;q,q}(\cV), \quad k,\ell,q,r \in \Z_+.
    \ee

The facts that $\left\{\varphi_{k,q}\right\}_{(k, q) \in \Z_+^2}$ is an orthonormal basis in $L^2(\rd)$, while $\left\{\varphi_{k,q}\right\}_{k \in \Z_+}$ is an orthonormal basis of ${\rm Ran}\,p_q$ with fixed $q \in \Z_+$, imply immediately the following elementary
 \beprl{np1}
 Let $\cV \in \gwrf$. \\
 {\rm (i)} The operator $\opwv$ is unitarily equivalent to  $\cM: \ell^2(\Z_+^2) \to \ell^2(\Z_+^2)$ defined by
    \bel{j22a}
    \left(\cM{\bf c}\right)_{k,q} : = \sum_{(\ell, r) \in \Z_+^2} m_{k,\ell;q,r}(\cV)\, c_{\ell,r}, \quad (k,q) \in \Z_+^2, \quad {\bf c} = \left\{c_{\ell,r}\right\}_{(\ell,r) \in \Z_+^2} \in \ell^2(\Z_+^2).\\
    \ee
 {\rm (ii)} Fix $q \in \Z_+$. Then the operator $\tqwv$ is unitarily equivalent to  $\cM_q: \ell^2(\Z_+) \to \ell^2(\Z_+)$ defined by
    \bel{j22}
    \left(\cM_q {\bf c}\right)_k : = \sum_{\ell \in \Z_+} m_{k,\ell;q}(\cV)\, c_\ell, \quad k \in \Z_+, \quad {\bf c} = \left\{c_\ell\right\}_{\ell \in \Z_+} \in \ell^2(\Z_+).
    \ee
   \epr
    We would like to give a more explicit form  of the matrices defining the operators $\cM$ and $\cM_q$, $q \in \Z_+$. To this end we need the following important
    \belel{jl1}
    We have
    \bel{j24}
    \cU_b^*\,\varphi_{k,q} = i^{q-k} \psi_q \otimes \psi_k, \quad k,q \in \Z_+,
    \ee
    where $\psi_q$, $q \in \Z_+$, are the Hermite functions  defined in \eqref{Dradi}.
    \ele
    \begin{proof}
    By \eqref{j25} -- \eqref{j27}, and \eqref{D134}, we get
    \bel{j28}
    \cU_b^*\,\varphi_{k,q} = \sqrt{\frac{b^{k+1}}{\pi 2^{k + q +1} k! q!}} ((\alpha^*)^q \otimes I_y)\,\cU_b^* \, u_k
    \ee
    where
    $$
    u_k(x,y) = (x+\I y)^k \E^{-b(x^2 + y^2)/4}, \quad (x,y) \in \rd.
    $$
    Using \eqref{D133}, we easily find that
    \bel{j30}
    \left(\cU_b^* u_k\right)(x,y) = \frac{1}{2\pi\sqrt{b}} \left(\frac{2}{\sqrt{b}}\right)^k \E^{\I xy} \left(\frac{\partial}{\partial \overline{z}}\right)^k \, J(x,y)
    \ee
    where
    $$
    J(x,y) : = \int_{\rd} \E^{-\I(ty- sx)} \, \E^{-\I ts/2} \, \E^{-(t^2+s^2)/4} \, \D t\,\D s, \quad (x,y) \in \rd.
    $$
    An elementary calculation yields
    \bel{j31}
     J(x,y)  = \sqrt{2} (2\pi) \E^{-\I xy} \, \E^{-(x^2+y^2)/2}.
     \ee
     Inserting \eqref{j31} into \eqref{j30}, we get
     \bel{j32}
     \left(\cU_b^* u_k\right)(x,y) = \sqrt\frac{2}{b^{k+1}}  \, \E^{-x^2/2} (-1)^k (\alpha^*)^k \E^{-y^2/2}.
     \ee
     Inserting \eqref{j32} into \eqref{j28}, we obtain \eqref{j24}.
     \end{proof}
     {\em Remark}: By \eqref{j24}, we have
     \bel{ja15}
     \cU_{b} f = \sum_{(k,q)  \in \Z_+^2} i^{k-q} \langle f, \psi_q \otimes \psi_k \rangle_{L^2(\rd)}\, \varphi_{k,q}, \quad f \in L^2(\rd).
    \ee
     Let $\cV \in \gwrf$. Set
    \bel{j19}
    v_{b,q}(y,\eta) : = \int_{\rd} \cV_b(x,y,\xi,\eta) \, \Psi_q(x,\xi) \, \D x\,\D \xi, \quad (y,\eta) \in \rd,
    \ee
   where  $\cV_b$ is the symbol defined in \eqref{j18}, and $\Psi_q$  is the Wigner function defined in \eqref{j15a}.

 \beprl{np2}  Let $\cV \in \gwrf$. Then we have
    $$
    m_{k,\ell;q,r}(\cV) = i^{k-\ell-q+r} \langle \cV_b, \Psi_{q,r} \otimes \Psi_{k,\ell} \rangle_{L^2(\re^4)}
    $$
    where
    $$
    \left(\Psi_{q,r} \otimes \Psi_{k,\ell}\right) (x,y,\xi,\eta) : = \Psi_{q,r}(x,\xi)\,\Psi_{k,\ell}(y,\eta), \quad k,\ell,q,r \in \Z_+, \quad (x,y,\xi, \eta) \in \re^4.
    $$
    In particular,
    \bel{n1}
     m_{k,\ell;q}(\cV) = i^{k-\ell} \langle \cV_b, \Psi_{q} \otimes \Psi_{k,\ell} \rangle_{L^2(\re^4)} = i^{k-\ell} \langle v_{b,q},  \Psi_{k,\ell} \rangle_{L^2(\re^2)}.
     \ee
     \epr
    \begin{proof} By  \eqref{j23}, \eqref{j4},  and \eqref{j24}, we have
    \begin{align}
    m_{k,\ell;q,r}(\cV) & = \langle \opw(\cV) \varphi_{\ell,r} , \varphi_{k,q}\rangle_{L^2(\rd)}\\ \nonumber
    & = \langle \cU_b^* \opw(\cV) \cU_b \, \cU_b^* \, \varphi_{\ell,r} , \cU_b^* \varphi_{k,q}\rangle_{L^2(\rd)}\\ \nonumber
    & = i^{k-\ell-q+r} \langle \opw(\cV_b)\, \psi_r \otimes \psi_\ell, \psi_q \otimes \psi_k \rangle_{L^2(\rd)}\\ \nonumber
    & = i^{k-\ell-q+r} \langle \cV_b, W(\psi_q\otimes \psi_k, \psi_r \otimes \psi_\ell) \rangle_{L^2(\re^4)}\\ \nonumber
    & = i^{k-\ell-q+r} \langle \cV_b, W(\psi_q, \psi_r) \otimes W(\psi_k, \psi_\ell) \rangle_{L^2(\re^4)}\\ \nonumber
    & = i^{k-\ell-q+r} \langle \cV_b, \Psi_{q,r} \otimes \Psi_{k,\ell} \rangle_{L^2(\re^4)}.
    \end{align}
    \end{proof}
    Let $q \in \Z_+$. By analogy with \eqref{ja15}, define the operator $\cU_{b,q} : L^2(\re) \to {\rm Ran}\,p_q$ by
    \bel{ja10}
    \cU_{b,q} f =  \sum_{k \in \Z_+} i^{k}\,\langle f, \psi_k\rangle_{L^2(\re)}\, \varphi_{k,q}, \quad f \in L^2(\re).
    \ee

    \becol{nf1}
     Let $q \in \Z_+$, $\cV \in \gwrf$. Then we have
     \bel{ja3}
     \cU_{b,q}^*\, \tqwv \,\cU_{b,q} =  \opw(v_{b,q}),
     \ee
     where $v_{b,q}$ is the symbol defined in \eqref{j19}.
      \eco
     \begin{proof}
     By \eqref{n1}, we have
     $$
     \langle \tqwv \varphi_{\ell,q}, \varphi_{k,q} \rangle_{L^2(\rd)} = m_{k,\ell;q}(\cV)
     = \I^{{k-\ell}} \langle v_{b,q} , \Psi_{k,\ell}\rangle_{L^2(\rd)} = \I^{{k-\ell}} \langle \opw(v_{b,q}) \psi_\ell, \psi_k\rangle_{L^2(\re)},
     $$
     which implies \eqref{ja3}.
    \end{proof}
     At the end of this section we consider the important case where the operator $\opw(v_{b,q})$ admits an anti-Wick symbol $\tilde{v}_{b,q} \in \gawrd$. Set
     \bel{n80}
     \omega_{b,q}(x,y) : = \tilde{v}_{b,q}(-b^{1/2}y, -b^{1/2}x), \quad (x,y) \in \rd.
     \ee
     Then, of course, $ \omega_{b,q} \in \gawrd$.
    \becol{nf10}
     Let $q \in \Z_+$,  $\cV \in \gwrf$. Assume that the operator $\opw(v_{b,q})$  has an anti-Wick symbol $\tilde{v}_{b,q} \in \gawrd$. Then, \bel{ja4} \cU_{b,0}^*\,p_0 \,\omega_{b,q} \,p_0\, \cU_{b,0} = \opaw(\tvbq) = \opw(v_{b,q}),
      \ee
      where $\cU_{b,0}$ is  the unitary operator defined in \eqref{ja10}, and $\omega_{b,q}$ is the symbol defined in \eqref{n80}.
     \eco
     \begin{proof}
      Assume at first that $\omega_{b,q} \in C_0^\infty(\rd)$. Then, by Corollary \ref{nf1}, the operator $p_0 \, \omega_{b,q} \, p_0 $ is unitarily equivalent under the operator $\cU_{b,0}$ to a $\Psi$DO with Weyl symbol
$$
    \int_{\rd} (\omega_{b,q} \circ \kappa_b)(x,y,\xi,\eta) \Psi_0(x,\xi) dx d\xi =
     \int_{\rd} \omega_{b,q}(b^{-1/2} (x-\eta), b^{-1/2}(\xi-y))
    \Psi_0(x,\xi) dx d\xi   =
    $$
    $$
    \frac{1}{\pi} \int_{\rd} \tilde{v}_{b,q} (y-\xi, \eta - x) e^{-(x^2+\xi^2)} dx d\xi =
    (\tilde{v}_{b,q} * \cG_1)(y,\eta), \quad (y,\eta) \in \rd,
    $$
    where we have taken into account \eqref{ja1}. Thus we get \eqref{ja4} for $\omega_{b,q} \in C_0^\infty(\rd)$.
    The result for general $\omega_{b,q} \in \gawrd$ is obtained by an approximation argument similar to the one applied in
the proof of \cite[Theorem 2.11]{PuRaVBl13}.
    \end{proof}
    The operator $p_0 \, \omega_{b,q} \, p_0$ admits a further useful unitary equivalence.  For $r \in \Z_+$ set
    \bel{d15}
    \cD_{b,r} : = {\rm L}_r\left(-\frac{\Delta}{2b}\right),
    \ee
    where  ${\rm L}_r$ is the Laguerre polynomial defined in \eqref{Dj3}. Thus, if $r=0$, we have $\cD_{b,0} = I$, and if $r \geq 1$, then $\cD_{b,r}$ is a partial differential operator with constant coefficients of order $2r$.

    \becol{df1}
      Assume that $\omega \in \gawrd$, and there exist $r \in \N$, $\zeta \in \cS'(\rd)$ such that
      \bel{d11a}
      \omega = \cD_{b,r} \, \zeta .
      \ee
      Then the operator $p_0 \, \omega \, p_0: {\rm Ran}\,p_0 \to {\rm Ran}\,p_0$ is unitarily equivalent to the operator $p_r \, \zeta  \, p_r: {\rm Ran}\,p_r \to {\rm Ran}\,p_r$.
      \eco
     \begin{proof} By \cite[Lemma 3.1]{BrPuRa04} and \eqref{d11a}, we have
     \bel{d11}
     \langle \omega \, \varphi_{k,0}, \varphi_{\ell,0}\rangle_{L^2(\rd)} = \langle (\cD_{b,r} \zeta) \, \varphi_{k,0}, \varphi_{\ell,0}\rangle_{L^2(\rd)}
    = \langle \zeta \, \varphi_{k,r}, \varphi_{\ell,r}\rangle_{L^2(\rd)}, \quad k,\ell \in \Z_+.
     \ee
     Let $u \in {\rm Ran}\,p_0$. Then $u = \sum_{k \in \Z_+} c_k \varphi_{k,0}$ with $\left\{c_k\right\}_{k \in \Z_+} \in \ell^2(\Z_+)$. Define the unitary operator $U_r: {\rm Ran}\,p_0 \to {\rm Ran}\,p_r$ by
     $U_r u : = \sum_{k \in \Z_+} c_k \varphi_{k,r}$. Then \eqref{d11} implies that
     $$
     p_0 \, \omega_{b,q} \, p_0 = U_r^*\,p_r \, \zeta \, p_r \, U_r.
     $$
     \end{proof}

     \section{Spectral properties of Weyl $\Psi$DOs with radial symbols}
     \label{s5}

     In this section we recall the fact that the Weyl $\Psi$DOs $\opwv$ with radial symmetric symbols $\cV$ are diagonalizable in the basis formed by  Hermite functions, and obtain explicit expressions for the eigenvalues of the operators $\opw(\cV)$ and $\opawv$. \\
     Let $n \geq 1$. We will say that the symbol $\cF \in \cS(\re^{2n})$ is {\em  radial}  if there exists a function $\cR_{\cF} : \re_+^n \to \C$ with $\re_+ : = [0,\infty)$, such that
     $$
     \cF(\bx, \bxi) = \cF(x_1,\ldots,x_n,\xi_1,\ldots,\xi_n) = \cR_{\cF}(x_1^2 + \xi_1^2, \ldots, x_n^2 + \xi_n^2), \quad (\bx, \bxi) \in \re^{2n}.
     $$
      We will say that $\cF \in \cS'(\re^{2n})$  is radial if for each $\cQ \in \cS(\re^{2n})$ there exists a radial symbol $\cR \in \cS(\re^{2n})$ such that
     $$
     (\cF, \cQ) = (\cF, \cR),
     $$
     $(\cdot, \cdot)$ being the usual pairing between $\cS$ and $\cS'$. Note that if $\cF \in \cS'(\re^n)$ is radial, then its Fourier transform $\widehat{\cF}$ is radial as well. Moreover, if the radial symbol $\cF$ is real-valued, then $\widehat{\cF}$ is real-valued as well.
    Set
     $$
     \cL_\bk(\bt) = \prod_{j=1}^n \left( {\rm L}_{k_j}(t_j) e^{-t_j/2}\right), \quad \bt = (t_1,\ldots,t_n) \in \re^n_+, \quad \quad \bk = (k_1,\ldots,k_n) \in \Z^n_+.
     $$
     As is well known, $\left\{\cL_\bk\right\}_{\bk \in \Z_+^n}$, and hence $\left\{(-1)^{|\bk|}\cL_\bk\right\}_{\bk \in \Z_+^n}$, are orthonormal bases in $L^2(\re^n_+)$. Note that the corresponding Fourier coefficients are defined not only for functions in $L^2(\re_+^n)$ but also for elements of  $L^1(\re_+^n) + L^\infty(\re_+^n)$, as well as for more general distributions (see e.g. \cite{Du90, JaPiPr17}).

    \beprl{np4}
     {\rm (i)} Let $\cF \in \gwrdn$ be a radial symbol. Then the operator $\opw(\cF)$ has eigenfunctions $\left\{\psi_\bk \right\}_{\bk \in \Z_+^n}$ with
     $$
     \psi_\bk(\bx) = \prod_{j=1}^n \psi_{k_j}(x_j), \quad \bx = (x_1,\ldots,x_n) \in \re^n, \quad \quad \bk = (k_1,\ldots,k_n) \in \Z^n_+,
     $$
where $\left\{\psi_q\right\}_{q \in \Z_+}$ are the Hermite functions defined in \eqref{Dradi}. The eigenfunctions $ \psi_\bk$ with $\bk \in \Z_+^n$
     correspond to  eigenvalues
    \begin{align} \label{n52}
     \mu_\bk^{\rm w}(\cF) & = \frac{(-1)^{|\bk|}}{2^n} \int_{\re_+^n} \cR_{\cF} (\bt/2) \, \cL_\bk(\bt)\,d\bt\\
      & = \int_{\re_+^n} \cR_{\widehat{\cF} }(2\bt) \, \cL_\bk(\bt)\,d\bt. \label{n56}
     \end{align}
     {\rm (ii)} Let $\cF : \Gamma_{\rm aw}(\re^{2n})$ be a radial anti-Wick symbol. Then the eigenfunctions $\left\{\psi_\bk\right\}_{\bk \in \Z_+^n}$
     of the operator $\opaw(\cF) = \opw(\cF * \cG_n)$ correspond to  eigenvalues
     \bel{n53}
     \mu_\bk^{\rm aw}(\cF)  =  \int_{\re_+^n} \cR_{\cF} (2\bt) \, \prod_{j=1}^n\left(\frac{t_j^{k_j} \,e^{-t_j}}{k_j !}\right)\,d\bt, \quad \bk \in \Z_+^n.
     \ee
     \epr
     {\em Remark}: In view of \eqref{1001}, it is not unnatural to express the  eigenvalues of $\opw(\cF)$ in terms of the Fourier transform $\widehat{\cF}$ of the symbol $\cF$, as in \eqref{n56}.
     \begin{proof}[Proof of Proposition \ref{np4}]
     We have
        \bel{n57}
        \langle \opw(\cF) \, \psi_\bl, \psi_\bk\rangle_{L^2(\re^{n})} = \langle \cF,  W(\psi_\bk, \psi_\bl)\rangle_{L^2(\re^{2n})}
        = \langle \cF, \otimes_{j=1}^n \Psi_{k_j, \ell_j} \rangle_{L^2(\re^{2n})}.
        \ee
        Due to the radial symmetry of $\cF$ and \eqref{n55}, we find that
        \bel{n58}
    \langle \cF, \otimes_{j=1}^n \Psi_{k_j, \ell_j}\rangle_{L^2(\re^{2n})} = \langle \cF, \otimes_{j=1}^n \Psi_{k_j}\rangle_{L^2(\re^{2n})} \prod_{j=1}^n  \delta_{k_j,\ell_j}.
    \ee
    By \eqref{n54},
    $$
    \langle \cF, \otimes_{j=1}^n \Psi_{k_j}\rangle_{L^2(\re^{2n})} =
    $$
    $$
    \frac{(-1)^{|\bk|}}{\pi^n} \int_{\re^{2n}} \cR_\cF(x_1^2 + \xi_1^2, \ldots, x_n^2 + \xi_n^2) \prod_{j=1}^n \left({\rm L}_{k_j}(2(x_j^2 + \xi_j^2))e^{-(x_j^2 + \xi_j^2)}\right)\, d\bx \, d\bxi.
    $$
    Changing the variables $x_j = r_j \cos{\theta_j}$, $\xi_j = r_j \sin{\theta_j}$, and then $t_j = 2 r_j^2$, $j=1,\ldots,n$, we obtain \eqref{n52}.
    In order to check \eqref{n56}, we first note that by the Parseval identity,
      $$
      \langle \cF, \otimes_{j=1}^n \Psi_{k_j}\rangle_{L^2(\re^{2n})} = \langle \widehat{\cF}, \otimes_{j=1}^n \widehat{\Psi}_{k_j}\rangle_{L^2(\re^{2n})}.
      $$
      By \cite[Eq. (3.6)]{PuRaVBl13}, we have
      \bel{d17}
      \widehat{\Psi}_k(\bw) = \frac{(-1)^k}{2} \Psi_k(\bw/2), \quad k \in \Z_+, \quad \bw \in \rd.
      \ee
      Therefore,
      $$
      \langle \widehat{\cF}, \otimes_{j=1}^n \widehat{\Psi}_{k_j}\rangle_{L^2(\re^{2n})} =
      $$
      $$
       \frac{1}{(2\pi)^n} \int_{\re^{2n}} \cR_{\widehat{\cF}}(x_1^2 + \xi_1^2, \ldots, x_n^2 + \xi_n^2) \prod_{j=1}^n \left({\rm L}_{k_j}((x_j^2 + \xi_j^2)/2)e^{-(x_j^2 + \xi_j^2)/4}\right)\, d\bx \, d\bxi,
    $$
    which implies \eqref{n56}.
    Let us now handle the anti-Wick case. Similarly to \eqref{n57} - \eqref{n58}, we have
    \bel{n60}
    \langle \opaw(\cF) \, \psi_\bl, \psi_\bk\rangle_{L^2(\re^{n})} = \langle \cF * \cG_n, \otimes_{j=1}^n \Psi_{k_j}\rangle_{L^2(\re^{2n})} \prod_{j=1}^n  \delta_{k_j,\ell_j}.
    \ee
     A simple calculation yields
    $$
     \langle \cF * \cG_n, \otimes_{j=1}^n \Psi_{k_j}\rangle_{L^2(\re^{2n})} =  \langle \cF, \cG_n *(\otimes_{j=1}^n \Psi_{k_j})\rangle_{L^2(\re^{2n})} =
    $$
    \bel{n61}
    \frac{(-1)^{|\bk|}}{4^n} \int_{\re_+^{n}} \cR_\cF(\btau/2) \prod_{j=1}^n \left(e^{-\tau_j/2} \int_0^\infty g(s\tau_j)\,{\rm L}_{k_j}(s)e^{-s}ds\right)\, d\tau \ee
    where
    $$
    g(y) : = (2\pi)^{-1} \int_0^{2\pi} e^{\sqrt{y} \, \cos{\theta}} d\theta, \quad y \geq 0.
    $$
    The function $g$ extends to an entire function satisfying
    $$
    g(z) = \sum_{j=0}^\infty \frac{1}{(j!)^2} \left(\frac{z}{4}\right)^j, \quad z \in \C.
    $$
    Using the first representation of the Laguerre polynomials in \eqref{Dj3}, we get
    \bel{n59}
    e^{-\tau/2} \,  \int_0^\infty g(s\tau)\,{\rm L}_{k}(s)e^{-s} ds = \frac{(-1)^k}{k!} e^{-\tau/4} \left(\frac{\tau}{4}\right)^k, \quad \tau \geq 0, \quad k \in \Z_+.
    \ee
    Inserting \eqref{n59} into \eqref{n61}, changing the variables $\btau = 4\bt$, and then inserting \eqref{n61} into \eqref{n60}, we get \eqref{n53}.
    \end{proof}

    {\em Remark}: In view of \eqref{n52}, \eqref{aa1} and \eqref{n54}, relation \eqref{n53} is equivalent to the fact that the Husimi function $\cG_1 * \Psi_k$
    can be written as
    \bel{d20}
   (\cG_1 * \Psi_k)(x,\xi) = \frac{1}{(2\pi)\,k!} \left(\frac{x^2+\xi^2}{2}\right)^k\,e^{-(x^2+\xi^2)/2}, \quad (x,\xi) \in \rd, \quad k \in \Z_+.
    \ee
    Probably, \eqref{d20} is known to the experts but since we could not find it in the literature, we include a proof of \eqref{n53}.

    \becol{nf50}
    Let $\cF \in \gwrdn$ be a radial symbol. \\
    {\rm (i)}
    Then  $\opw(\cF) \geq 0$ if and only if the Fourier coefficients
    of the function $ \cR_{\cF}( \bt/2)$, $\bt \in \re_+^n$, with respect to the system $\left\{(-1)^{\bk|}\cL_\bk\right\}_{\bk \in \Z_+^n}$, are non-negative.\\
    {\rm (ii)} Equivalently, we have  $\opw(\cF) \geq 0$ if and only if the Fourier coefficients
    of the function $\cR_{\widehat{\cF}}(2 \bt)$, $\bt \in \re_+^n$, with respect to the system $\left\{\cL_\bk\right\}_{\bk \in \Z_+^n}$ are non-negative.
    \eco
    \begin{proof}
    The first part  follows from \eqref{n52}, and the second one from \eqref{n56}.
    \end{proof}
    {\em Remark}: The criterion in the first part of Corollary \ref{nf50} has been established in \cite{Un16} for the one-dimensional case $n=1$, and in  \cite{JaPiPr17} for the multidimensional case. Presumably, at heuristic level, these facts have been known since long ago.
    \becol{nf51}
     Let $\cF \in \gawrdn$ be a radial symbol. Then  $\opaw(\cF) \geq 0$ if and only if
    $$
    \int_{\re_+^n} \cR_{\cF} (2\bt) \, \prod_{j=1}^n\left(t_j^{k_j} \,e^{-t_j}\right)\,d\bt \geq 0, \quad \bk \in \Z_+^n.
    $$
    \eco
    \begin{proof}
    The claim  follows from \eqref{n53}.
    \end{proof}

    In the case $n=1$, Proposition \ref{np4} tells us that the matrix $\left\{\langle \opw(\cF) \psi_\ell, \psi_j \rangle_{L^2(\re)}\right\}_{j,\ell \in \Z_+}$ is diagonal, provided that the symbol $\cF$ is radial. This fact admits an obvious generalization to the case where $\cF(r\cos{\theta}, r\sin{\theta})$ has a finite
    Fourier series with respect to the angle $\theta$.

     \beprl{jf3} Let $\cF \in \gwrd$. Assume that
     there exists $K \in \Z_+$ such that
     $$
     \cF (r\cos{\theta},r\sin{\theta}) = \sum_{k = -K}^K \cF_k(r) \E^{\I k\theta}, \quad r \in [0,\infty), \quad \theta \in [0,2\pi).
     $$
     Then the matrix $\left\{\langle \opw(\cF) \psi_\ell, \psi_j\rangle_{L^2(\re)}\right\}_{j,\ell \in \Z_+}$ is $(2K+1)$-diagonal.
     \epr
     Of course, Proposition \ref{jf3} admits an obvious extension to any dimension $n \geq 1$.\\
    %In view of Corollary \ref{jf3}, the results and the methods of \cite{Pu18} may turn out to be useful.
    Proposition \ref{np4} allows us to  calculate explicitly the spectrum of the perturbed Landau Hamiltonian $H_\cV = H_0 + \opw(\cV)$ provided that the symbol $\cV_b$ is radial.

    \becol{nf8}
    Let $\cV \in \gwrf$. Assume that the symbol $\cV_b = \cV \circ \kappa_b$ is radial. Then the operator $H_\cV$, normal on the domain $\gD(H_0)$, has eigenfunctions $\left\{\varphi_{k,q}\right\}_{(k,q) \in \Z_+^2}$ which correspond to eigenvalues
    $$
    \Lambda_q + \mu^{\rm w}_{(q,k)}(\cV_b), \quad (q,k) \in \Z_+^2.
    $$
    \eco
    \begin{proof}
    By Proposition \ref{jp1} we have
    $$
    H_\cV = \cU_b\,(((b\gh) \otimes I_y) + \opw(\cV_b))\,\cU_b^*,
    $$
    while Lemma \ref{jl1} and Proposition \ref{np4} imply
    $$
    \cU_b\,(((b\gh) \otimes I_y) + \opw(\cV_b))\,\cU_b^*\,\varphi_{k,q} = (\Lambda_q + \mu^{\rm w}_{(q,k)}(\cV_b))\,\varphi_{k,q}, \quad (k,q) \in \Z_+^2.
    $$
    \end{proof}

     \section{Eigenvalue distribution for the operator $H_\cV$}
     \label{s4}
 \subsection{Main results}
    \label{ss4}
In this section we study the eigenvalue asymptotics  near a fixed Landau level $\Lambda_q$, $q \in \Z_+$, of the perturbed Landau Hamiltonian $H_\cV = H_0 + \opwv$ with appropriate symbol $\cV$ such that $\opwv$ is bounded, self-adjoint  in $L^2(\rd)$, and relatively compact with respect to $H_0$.

Proposition \ref{np51} below shows, in particular, that the eigenvalues of $H_{-\cV}$ with $\opwv \geq 0$ and  $\cV \in \cS(\re^4)$, adjoining the Landau levels $\Lambda_q$, $q \in \Z_+$, may have quite arbitrary asymptotic  behavior; they may not accumulate at a given $\Lambda_q$, or may accumulate at any prescribed sufficiently fast accumulation rate. \\
Let $T$ be an operator, self-adjoint in a given Hilbert space, and $(\mu_1,\mu_2)$ be an open interval with $-\infty \leq \mu_1 < \mu_2 \leq \infty$. Set
$$
N_{(\mu_1,\mu_2)}(T) : = {\rm Tr}\,\one_{(\mu_1,\mu_2)}(T).
$$
Here and in the sequel $\one_S$ denotes the characteristic function of the set $S$. Thus, $\one_{(\mu_1,\mu_2)}(T)$ is just the spectral projection of $T$ corresponding to the interval $(\mu_1,\mu_2)$. If $(\mu_1,\mu_2) \cap \sigma_{\rm ess}(T) = \emptyset$, then $N_{(\mu_1,\mu_2)}(T)$ is the number of the
eigenvalues of $T$, lying on $(\mu_1,\mu_2)$ and counted with the multiplicities. If, moreover, $\overline{(\mu_1,\mu_2)} \cap \sigma_{\rm ess}(T) = \emptyset$, then $N_{(\mu_1,\mu_2)}(T) < \infty$. \\

    \beprl{np51}
    Let $\left\{m_q\right\}_{q \in \Z_+}$ be a given sequence with $m_q \in \Z_+ \cup \{\infty\}$, $q \in \Z_+$. Then there exists a symbol $\cV \in \cS(\re^4)$ such that $\cV \circ \kappa_b$ is radial, $\opwv \geq 0$, and
    \bel{n75}
    N_{I_q^-}(H_{-\cV}) = m_q, \quad q \in \Z_+,
    \ee
    where $I_q^-$ are the intervals defined in \eqref{fin30}.
    \epr
     \begin{proof}
    %[Proof of Proposition \ref{np51}]
    Set
    $$
    \cZ : = \left\{q \in \Z_+ \, | \, m_q \neq 0\right\}.
    $$
    If $\cZ = \emptyset$, it suffices to take $\cV = 0$. Assume $\cZ \neq \emptyset$. Let $\left\{c_{1,q}\right\}_{q \in \cZ}$ be a decreasing set of numbers $c_{1,q} \in (0,2b)$; if $0 \in \cZ$, we can omit the condition $c_{1,0} < 2b$. If $\# \cZ = \infty$, we assume that $\lim_{q \to \infty} q^m c_{1,q} = 0$ for any $m \in \N$. Fix $q \in \cZ$. Let $\left\{c_{2,k}\right\}_{k = 0}^{m_q-1}$ be a decreasing set of numbers $c_{2,k} \in (0,1)$. If $m_q = \infty$, we assume that $\lim_{k \to \infty} k^m c_{2,k} = 0$ for any $m \in \N$.
    Now put
    $$
    C_{k,q} : =  c_{1,q} c_{2,k}, \quad k=0,\ldots, m_q-1, \quad q \in \cZ,
    $$
    \bel{d22}
    \cV : = (2\pi)^2 \left(\sum_{q \in \cZ} \sum_{k=0}^{m_q-1} C_{k,q} \Psi_q \otimes \Psi_k\right) \circ \kappa_b^{-1}.
    \ee
    Then,  $\cV \in \cS(\re^4)$ (see \cite[Theorem 2.5 (a)]{Du90}), and, evidently, $\cV \circ \kappa_b$ is  radial. Moreover, by Corollary \ref{nf8}, $\opwv \geq  0$ and
    \bel{d23}
    \sigma(H_{-\cV}) \cap I_q^- = \left\{
    \begin{array} {l}
    \emptyset \quad {\rm if} \quad q \not \in \cZ,\\ [2mm]
    \cup_{k=0}^{m_q-1} \left\{\Lambda_q - C_{k,q}\right\} \quad {\rm if} \quad q \in \cZ.
    \end{array}
    \right.
    \ee
    By construction, all the eigenvalues $\Lambda_q - C_{k,q}$, $k = 0,\ldots,m_q-1$, lying in $I_q^-$ with $q \in \cZ$, are simple.
    Therefore, \eqref{n75} holds true.
    \end{proof}
    {\em Remarks}: (i) The proof of Proposition \ref{np51} contains an explicit construction of a negative compact perturbation of $H_0$ so that the eigenvalues $H_{-\cV}$ may accumulate to $\Lambda_q$ only from below. Of course, it is possible to construct positive compact perturbations whose eigenvalues may accumulate to $\Lambda_q$ only from above, or self-adjoint compact perturbation with non-trivial positive and negative parts whose eigenvalues may accumulate to $\Lambda_q$ both from above and from below. \\
    (ii) It is easy to check that if for some $q \in \Z_+$ we have $m_q < \infty$, then the Landau level remains an eigenvalue of infinite multiplicity of $H_{-\cV}$. In contrast to this situation, it was shown in \cite{KlRa09} that if $\opw(V) = V$ is local, i.e. if $V = V(x,y)$, $(x,y) \in \rd$, and $V \leq 0$, $\|V\|_{L^\infty(\rd)} < 2b$, then ${\rm Ker}\, (H_{\pm V} - \Lambda_q I) = \{0\}$. \\
    (iii) It is an elementary fact that if ${\bf c} : = \left\{c_{k,q}\right\}_{(k,q) \in \Z_+^2} \in \ell^\infty(\Z_+^2)$, then the operator
    \bel{fin3}
    T : = \sum_{(k,q) \in \Z_+^2} c_{k,q} \langle\cdot,\varphi_{k,q}\rangle \varphi_{k,q}
    \ee
    is bounded in $L^2(\rd)$, $\|T\| = \|{\bf c}\|_{\ell^\infty(\Z_+^2)}$, the eigenvalues of $T$ coincide with the set $\left\{c_{k,q}\right\}_{(k,q) \in \Z_+^2}$,  while the eigenvalues of $H_0 + T$ coincide with $\left\{\Lambda_q + c_{k,q}\right\}_{(k,q) \in \Z_+^2}$. Choosing appropriately the sequence ${\bf c}$, we can easily obtain operators with various spectral properties. If, for example, $\left\{r_k\right\}_{k \in \Z_+}$ is the set of the rational numbers on $(0,2b)$, and
    $$
    c_{k,q} : = r_k, \quad (k,q) \in \Z_+^2,
    $$
    then the operator $H_0 + T$ will have purely dense point spectrum $\sigma(H_0+T) = [b,\infty)$. Of course, in this case $T$ is not relatively compact with respect to $H_0$. \\
    However, in the general case it would not be possible to interpret $T$ as a Weyl $\Psi$DO with a regular symbol. Our assumption in Proposition \ref{np51}  that the sequence $\left\{C_{k,q}\right\}$ decays  rapidly implies that the symbol $\cV$ defined in \eqref{d22} belongs to the class $\cS(\re^4)$. If, for example, we assume instead that we have only
    $$
    \sum_{q \in \cZ} \sum_{k=0}^{m_q-1} c_{1,q}^2 c_{2,k}^2 < \infty,
    $$
    then the symbol defined in \eqref{d22} generates by Proposition \ref{Dprvbp1c} a Hilbert-Schmidt operator. In this case, \eqref{d23} still holds true, just the eigenvalues of $H_{-\cV}$ lying in a given gap $I_q^-$ may accumulate more slowly to $\Lambda_q$.\\

   In the following two theorems we assume that the operator $\opwv$ satisfies two general assumptions:\\

   $\bf{H_1}$: The operator $\opwv$ is bounded and self-adjoint, and the operator $\opwv H_0^{-1}$ is compact in $L^2(\rd)$. Moreover, $\opwv \geq 0$.\\

    $\bf{H_{2,q,r}}$: Let $q \in \Z_+$. Then the operator $\opw(v_{b,q})$, $v_{b,q}$ being defined in \eqref{j19}, has an anti-Wick symbol $\tilde{v}_{b,q} \in \gawrd$. Moreover, there exists $r \in \Z_+$ and $0 \leq \zeta_{b,q,r} \in L^\infty(\rd)$ such that \eqref{d16} holds true, i.e we have
     $$
     \omega_{b,q} = \cD_{b,r}\,\zeta_{b,q,r}
     $$
    where $\omega_{b,q}$ is the symbol defined in \eqref{n80}, and $\cD_{b,r}$ is the differential operator defined in \eqref{d15}.  \\

    {\em Remarks}: (i) In what follows we will write $\zeta$ instead of $\zeta_{b,q,r}$.\\
    (ii) It is easy to check that for each $q,r \in \Z_+$ there exist symbols $\cV \in \gwrf$ satisfying Assumptions $\bf{H_{1}}$ and $\bf{H_{2,q,r}}$. A simple example can be constructed as follows. Pick $0 \leq \zeta \in C^\infty(\rd)$, bounded together with all its derivatives. Set $\omega = \cD_{b, r}\,\zeta$,
    $\tilde{v}(x,y) : = \omega(-b^{-1/2}y, b^{-1/2} x)$, $(x,y) \in \rd$, $v : = \tilde{v} * \cG_1$, and
    $$
    \cV : = 2\pi\left(\Psi_q \otimes v\right) \circ \kappa_b^{-1}.
    $$
    Then, according to \eqref{j19}, we have $v_{b,q} = v$, and hence $\cV$ satisfies $\bf{H_{1}}$ and $\bf{H_{2,q,r}}$. However, if we consider the operator defined in \eqref{d22}, and assume that for a certain $q \in \Z_+$ we have $0 < m_q < \infty$, then the corresponding $\opw(v_{b,q})$ {\rm does not admit} an anti-Wick symbol $\tilde{v}_{b,q} \in \gawrd$. Indeed, in this case we have
    $$
    \opw(v_{b,q}) = 2\pi \sum_{k=0}^{m_q-1} C_{k,q} \Psi_k.
    $$
    If $v_{b,q} = \tilde{v}_{b,q} * \cG_1$ with $\tilde{v}_{b,q} \in \cS'(\rd)$, then \eqref{d17} easily implies that the Fourier transform of $\tilde{v}_{b,q}$ is a polynomial so that $\tilde{v}_{b,q} \in \cE'(\rd)$ with ${\rm supp}\,(\tilde{v}_{b,q}) = \{0\}$.\\
    (iii) As we will see in the proof of Theorem \ref{np7} below, the Toeplitz operator $\tqw(\cV)$ is the effective Hamiltonian which governs the eigenvalue asymptotics of  $H_{\pm \cV}$ near the Landau level $\Lambda_q$, $q \in \Z$.  The operator $\tqw(\cV) = p_q \opwv p_q$ is an appropriate restriction of
    non-local $\Psi$DO $\opwv$, and is unitarily equivalent by Corollary \ref{nf1} to $\opw(v_{b,q})$. By our assumption, $\opw(v_{b,q})$ admits an anti-Wick symbol $\tilde{v}_{b,q}$ and, hence, by Corollary \ref{nf10} it is unitarily equivalent to $p_0 \, \omega_{b,q} \, p_0$, a restriction of the {\em local} multiplier $\omega_{b,q}$. Thus, the existence of an anti-Wick symbol $\tilde{v}_{b,q}$ of $\opw(v_{b,q})$ allows us to replace, in a certain sense, the non-local operator $\opw(\cV)$ by the local one $\omega_{b,q}$ in the asymptotic analysis of the  eigenvalue distribution of $H_{\pm\cV}$ near $\Lambda_q$. As mentioned in the Introduction, similar substitutions of non-local potentials by local ones have been considered in the physics literature (see e.g.
    \cite{CoArMa70, SoKi97, ChMoKiChKiSo14}).\\
     (iv) We introduce the passage from $\omega_{b,q}$ to $\zeta$ in \eqref{d16} in particular due to our requirement that $\zeta$ is non-negative: it may happen that $\zeta \geq 0$ while $\omega_{b,q}$ is not sign-definite. \\

     In Theorem \ref{nth1} (resp., Theorem \ref{nth2}) below we  study the eigenvalue asymptotics for the operators $H_{\pm \cV}$ at a given Landau level $\Lambda_q$, $q \in \cZ_+$, under Assumptions  $\bf{H_1}$ and $\bf{H_{2,q,r}}$, supposing in addition that $\zeta$ is compactly supported  (resp., that $\zeta$ decays exponentially at infinity).  Since $\opwv \geq 0$ by Assumption  $\bf{H_1}$, the eigenvalues of the operator $H_\cV$ may accumulate to given Landau level $\Lambda_q$ only from above, while the eigenvalues of $H_{-\cV}$ may accumulate to $\Lambda_q$ only from below, as mentioned in the Introduction,. We recall the notations $\left\{\lambda^\pm_{k,q}(\cV)\right\}_{k=0}^{m_q^\pm-1}$ of the operator $H_{\pm \cV}$ lying on the interval $I_q^\pm$, $q \in \Z_+$.

    For the formulation of our first theorem we need  the notion of {\em a logarithmic capacity} $\gC(K)$ of a compact set $K \subset \rd$ (see e.g. \cite[Chapter 5]{Ran95}).
    Let $M(K)$ denote the set of probability measures on $K$. Then we have
$\gC(K) : = e^{-{\mathcal I}(K)}$ where
$$
{\mathcal I}(K) : = \inf_{\mu \in M(K)}  \int_{K \times K} \ln{|x-y|^{-1}} d\mu(x) d\mu(y).
$$
If $K_1 \subset K_2$, then, evidently, $\gC(K_1) \leq \gC(K_2)$.\\
%If $K$ is simply connected, then $\gC(K)$ is equal to the {\em conformal radius} of $K$ (see e.g. \cite[Chapter II, Section 4]{La72}).
%Under the assumption that ${\rm supp}\,(\tilde{v}_{b,q})$ is compact, set
%$\gC_{b,q} : = \gC({\rm supp}\,(\tilde{v}_{b,q}))$.

    \begin{theorem} \label{nth1}
    Let $\bf{H_{1}}$ and $\bf{H_{2,q,r}}$ with fixed $q, r \in \Z_+$,  hold true.
    Assume that $\zeta \in C(\rd)$, ${\rm supp}\,\zeta = \overline{\Omega}$ where $\Omega \subset \rd$ is a bounded domain with Lipschitz boundary $\partial \Omega$, and $\zeta > 0$ on $\Omega$.  Then $m_q^\pm = \infty$, and we have
    \bel{1}
    \ln{\left(\pm\left(\lambda_{k,q}^\pm(\cV) - \Lambda_q\right)\right)} = -k\ln{k} + \left(1 + \ln{\left(\frac{b\,\gC(\overline{\Omega})^2}{2}\right)}\right)k + o(k), \quad k \to \infty.
    \ee
    \end{theorem}

    {\em Remarks}: %\left\{C_{k,q}\right\}_{(k,q) \in \Z_+^2}
    (i) Assume that $\zeta \in L^\infty(\rd)$, ${\rm supp}\,\zeta$ is compact and for some $C>0$, $r>0$,  and ${\bf x}_0 \in \rd$ we have
    $\zeta({\bf x}) \geq C \one_{B_r({\bf x}_0)}({\bf x})$ where $B_r({\bf x}_0) : = \left\{{\bf x} \in \rd \, | \, |{\bf x} - {\bf x}_0| < r\right\}$.
    Then \cite{RaWa02} implies
     \bel{7}
     \ln{\left(\pm\left(\lambda_{k,q}^\pm(\cV) - \Lambda_q\right)\right)} = -k\ln{k}\,(1 + o(1)), \quad k \to \infty.
     \ee
    which is a less precise version of \eqref{1}. \\
    (ii) By \cite{Te19}, Theorem \ref{nth1} is valid under more general assumptions on ${\rm supp}\,\zeta = \overline{\Omega}$. Namely, we can suppose that there exists a compact set $Z \subset \Omega$ such that $\zeta > 0$ only on $\Omega \setminus Z$ and not on the entire domain $\Omega$. We omit the details of the proof of this extension for the sake of the simplicity of the exposition.\\

Our next theorem  concerns the case where $\zeta$ decays exponentially at infinity. Now we assume  that $\zeta \in C(\rd)$ and there exist $\beta>0$ and $\gamma > 0$ such that
     \bel{3}
     \ln{\zeta(\bx)} = - \gamma |\bx|^{2\beta} + \cO(\ln{|\bx|}), \quad |\bx| \to \infty,
     \ee
    uniformly with respect to $\frac{\bx}{|\bx|} \in {\mathbb S}^1$. Set $\mu : = \gamma (2/b)^\beta $ where $b>0$ is the constant scalar magnetic field.

     \begin{theorem} \label{nth2} Let $\bf{H_{1}}$ and $\bf{H_{2,q,r}}$ with fixed $q, r \in \Z_+$,  hold true.
    Assume
     that $\zeta$ satisfies \eqref{3}. Then $m_q^\pm = \infty$ and we have: \\
     {\rm (i)} If $\beta \in (0,1)$, then there exist constants $f_j = f_j(\beta, \mu)$, $j \in \N$, with $f_1 = \mu$, such that
     \bel{4}
     \ln{\left(\pm\left(\lambda_{k,q}^\pm(\cV) - \Lambda_q\right)\right)} = - \sum_{1 \leq j < \frac{1}{1-\beta}} f_j k^{(\beta-1)j + 1} + \cO(\ln{k}), \quad k \to \infty.
     \ee
     {\rm (ii)} If $\beta = 1$, then
     \bel{5}
     \ln{\left(\pm\left(\lambda_{k,q}^\pm(\cV) - \Lambda_q\right)\right)} = - \left(\ln{(1+\mu)}\right) k + \cO(\ln{k}), \quad k \to \infty.
     \ee
     {\rm (iii)} If $\beta \in (1,\infty)$, then there exist constants $g_j = g_j(\beta, \mu)$, $j \in \N$,  such that
     $$
     \ln{\left(\pm\left(\lambda_{k,q}^\pm(\cV) - \Lambda_q\right)\right)} =
     $$
     \bel{6}
    - \frac{\beta - 1}{\beta} k \ln{k} + \left(\frac{\beta - 1 - \ln{(\mu\beta)}}{\beta}\right) k - \sum_{1 \leq j < \frac{\beta}{\beta-1}} g_j k^{(\frac{1}{\beta}-1)j + 1} + \cO(\ln{k}), \quad k \to \infty.
     \ee
      \end{theorem}
      The coefficients $f_j$ and $g_j$, $j \in \N$, appearing in \eqref{4} and \eqref{6}, are described explicitly in
      \cite[Theorem 2.2]{LuRa15}. For the completeness of the exposition, we reproduce this description here.
      Assume at first $\beta \in (0,1)$. For $s>0$ and $\epsilon \in \re$, $|\epsilon| \ll 1$, introduce the function
      %\bel{o1a}
      $$
      F(s;\epsilon) : = s-\ln{s} + \epsilon \mu  s^\beta.
      $$
      %\ee
      Denote by $s_<(\epsilon)$ the unique positive solution of the equation $s = 1 - \epsilon \beta \mu s^\beta$, so that $\frac{\partial F}{\partial s}(s_<(\epsilon); \epsilon) = 0$. Set
      %\bel{o4a}
      $$
      f(\epsilon) : = F(s_<(\epsilon); \epsilon).
      $$
      %\ee
      Note that $f$ is a real analytic function for small $|\epsilon|$. Then $f_j: = \frac{1}{j!} \frac{d^j f}{d\epsilon^j}(0)$, $j \in \N$. \\
      Let now $\beta \in (1,\infty)$. For $s>0$ and $\epsilon \in \re$, $|\epsilon| \ll 1$, introduce the function
      %\bel{o9a}
      $$
      G(s;\epsilon) : = \mu s^\beta -\ln{s} + \epsilon s.
      $$
      %\ee
      Denote by $s_>(\epsilon)$ the unique positive solution of the equation $\beta \mu s^\beta = 1 - \epsilon s$ so that $\frac{\partial G}{\partial s}(s_>(\epsilon); \epsilon) = 0$. Define
      %\bel{o10a}
      $$
      g(\epsilon) : = G(s_>(\epsilon); \epsilon),
      $$
      %\ee
      which is a real analytic function for small $|\epsilon|$. Then $g_j: = \frac{1}{j!} \frac{d^j g}{d\epsilon^j}(0)$, $j \in \N$.\\

      %(ii) If we assume, instead of \eqref{3}, that $\zeta \in C(\rd)$ and
    % $$
     %\ln{\zeta(\bx)} = - \gamma |\bx|^{2\beta}(1 + o(1)), \quad |\bx| \to \infty,
   % $$
     % then, similarly to \eqref{7}, \cite{RaWa04} implies
    %  $$
    %\ln{\left(\pm\left(\lambda_{k,q}^\pm - \Lambda_q\right)\right)}   =  \left\{
     % \begin{array} {l}
     % -\mu k^{\beta} (1 + o(1)) \quad {\rm if} \quad \beta \in (0,1),\\[2mm]
      % -(\ln{(1+\mu)})\, k (1 + o(1)) \quad {\rm if} \quad \beta = 1,\\[2mm]
     % -\frac{\beta-1}{\beta} k \, \ln{k}(1 + o(1))\quad {\rm if} \quad \beta \in (1,\infty),
      % \end{array}
      %\right.
      %\quad k \to \infty.
      %$$
      %which is a less precise version \eqref{4} -- \eqref{6}.
      In our next theorem we deal with the case where $v_{b,q}$ admits a power-like decay at infinity. Our general assumption concerning the perturbation $\opwv$ is: \\

      $\bf{H_3}$ The symbol $\cV$ is real-valued and satisfies the hypotheses of Proposition \ref{np90}.\\

      We recall that under Assumption $\bf{H_3}$ the operator $\opwv$ is self-adjoint and bounded in $L^2(\rd)$, and $\opwv H_0^{-1}$ is compact. However, we do not suppose now that $\opwv$ has a definite sign.
      Further, under Assumption $\bf{H_3}$, there exists a symbol $\cW \in \cS_\varrho^{0}(\re^4)$ such that $\opwv^2 = \opw(\cW)$.
    By analogy with \eqref{j19}, set
    $$
    w_{b,q}(y,\eta) : = \int_{\rd} (\cW \circ \kappa_b)(x,\xi,y,\eta) \Psi_q(x,\xi) dx d\xi,  \quad (y,\eta) \in \rd, \quad q \in \Z_+.
    $$
    Our next assumption concerns the decay of the symbols $v_{b,q}$ and $w_{b,q}$ at infinity:\\

    $\bf{H_{4,q,\gamma}}$ Let $q \in \Z_+$. Then there exist $\gamma > 0$ and $\varrho \in (0,1]$ such that $v_{b,q} \in \cS_\varrho^{-\gamma}(\rd)$ and $w_{b,q} \in \cS_\varrho^{-2\gamma}(\rd)$.\\

    {\em Remark}: A simple sufficient condition which guarantees the fulfillment of $\bf{H_3}$ and $\bf{H_{4,q,\gamma}}$ is that $\cV \in \cS_\varrho^{-\gamma}(\re^4)$ with some $\gamma > 0$ and $\varrho \in (0,1]$. In this case, the operator $\opwv$ is not only bounded but also compact.
    Another condition which implies the validity of $\bf{H_3}$ and $\bf{H_{4,q,\gamma}}$ is that $\opw(V) = V$ is a local potential, and $V \in \cS_\varrho^{-\gamma}(\re^2)$. This case corresponds to an electric perturbation of $H_0$ and was considered in \cite{Ra90, Iv98}.\\

       It is more convenient to formulate Theorem \ref{oth1} below  in the terms of
     eigenvalue counting functions. For
     $S = S^* \in \gB(L^2(\rd))$, $S H_0^{-1} \in \gS_\infty(L^2(\rd))$,
     set
    %\bel{o1}
    $$
    \cN_q^>(\lambda; S) : = N_{(\Lambda_q+\lambda, \Lambda_q + b)}(H_0+S), \quad \lambda \in (0,b), \quad q \in \Z_+,
    $$
    %\ee
    %and for $q \in \N$ put
    %\bel{o2}
    $$
    \cN_q^<(\lambda; S) : = N_{(\Lambda_{q-1}+b,  \Lambda_q -\lambda)}(H_0+S), \quad \lambda \in (0,b), \quad q \in \N,
    $$
    %\ee
    $$
    \cN_0^<(\lambda; S) : = N_{(-\infty,  \Lambda_q -\lambda)}(H_0+S), \quad \lambda >0.
    $$

  Let $f : (0,\infty) \to [0,\infty)$ be a non-increasing function. We will say that $f$ satisfies the condition $\cC$ if there exists $\lambda_0 \in (0,\infty)$ such that:
    \begin{itemize}
    \item $f$ is derivable on $(0,\lambda_0)$;
    \item there exist numbers $0 < \gamma_1 < \gamma_2<\infty$ such that for any $\lambda \in (0,\lambda_0)$ we
    have
    \bel{o10}
    \gamma_1 f(\lambda) < - \lambda f'(\lambda) < \gamma_2 f(\lambda).
    \ee
    \end{itemize}
    Let $n \in \N$. For a Lebesgue-measurable function $\cF : \re^{2n} \to \re$ set
    $$
    \gV_n^\pm(\lambda; \cF) : = (2\pi)^{-n} \left|\left\{ (x,\xi) \in \re^{2n} \, | \, \pm \cF(x,\xi) > \lambda\right\}\right|, \quad \lambda > 0,
    $$
    where $| \cdot |$ is the Lebesgue measure in $\re^{2n}$.

    \begin{theorem} \label{oth1}
     Assume that $\cV$ satisfies $\bf{H_3}$ and $\bf{H_{4,q,\gamma}}$ with $q \in \Z_+$ and $\gamma >0$. Assume that the functions $\gV_1^\pm(\cdot;v_{b,q})$, $v_{b,q}$ being defined in \eqref{j19}, satisfy the condition $\cC$. If $\liminf_{\lambda \downarrow 0} \lambda^{2/\gamma} \gV_1^+(\lambda; v_{b,q}) > 0$ {\rm (}resp., if $\liminf_{\lambda \downarrow 0} \lambda^{2/\gamma} \gV_1^-(\lambda; v_{b,q}) > 0${\rm )}, then we have
    \bel{o11}
    \cN_q^>(\lambda; \opw(\cV)) = \gV_1^+(\lambda; v_{b,q}) (1 + o(1)), \quad \lambda \downarrow 0,
    \ee
    or, respectively,
    \bel{o12}
    \cN_q^<(\lambda; \opw(\cV)) = \gV_1^-(\lambda; v_{b,q}) (1 + o(1)), \quad \lambda \downarrow 0.
    \ee
    \end{theorem}

    {\em Remark}: It is easy to show that there exists $\delta>0$ such that we can replace $o(1)$ by $\cO(\lambda^\delta)$ in the remainder estimates in \eqref{o11} - \eqref{o12}. Since anyway these remainder estimates would not be sharp, we omit the tedious technical details.

     \subsection{Proofs of Theorems \ref{nth1} and \ref{nth2}}
    \label{ss42}
    \subsubsection{Auxiliary results}
    Let $T=T^*$ be a compact operator in a Hilbert space $X$. For $s>0$ set
    $$
    n_\pm(s;T) : = N_{(s,\infty)}(\pm T).
    $$
    If ${\rm rank}\,T_+ = \infty$ which is equivalent to $\lim_{s \downarrow 0} n_+(s;T) = \infty$, denote by $\left\{\nu_k(T)\right\}_{k=0}^{\infty}$ the non-increasing sequence of the positive eigenvalues of $T$.
    If $T_j = T_j^* \in \gS_\infty(X)$ and $s_j > 0$, $j=1,2$, then {\em the Weyl inequalities}
    \bel{o30}
    n_\pm(s_1 + s_2; T_1 + T_2) \leq n_\pm(s_1;T_1) + n_\pm(s_2;T_2)
    \ee
    hold true  (see e.g. \cite[Theorem 9, Section 9.2]{BiSo87}).
    \begin{proposition} \label{np7}
     Suppose that $\bf{H_1}$ and $\bf{H_{2,q,r}}$ with $q, r \in \Z_+$, hold true and  $\zeta$ satisfies the assumptions of Theorem {\rm \ref{nth1}} or of Theorem {\rm \ref{nth2}}. Then ${\rm rank}\,(p_r \zeta p_r) = \infty$, $m_q^\pm = \infty$, and  for each $\veps \in (0,1)$ there exists $k_0 \in \Z_+$ such that for sufficiently large $k \in \N$ we have
    \bel{n60a}
   \frac{1}{1+\veps} \nu_{k+k_0}(p_r \zeta p_r)  \leq
   \pm (\lambda_{k,q}^\pm(\cV) - \Lambda_q) \leq
    \frac{1}{1-\veps} \nu_{k-k_0}(p_r \zeta p_r).
    \ee
    \epr
    \begin{proof}
    %\left\{C_{k,q}\right\}_{(k,q) \in \Z_+^2}
    By the generalized Birman-Schwinger principle (see e.g. \cite[Theorem 1.3]{AlAvDeHe94}, \cite[Proposition 1.6]{Bi91}),
    \bel{n61a}
    \cN_q(\lambda; \pm \opwv) = n_\mp(1; \opwv^{1/2} (H_0 - \Lambda_q \mp \lambda)^{-1} \opwv^{1/2}) + \cO(1).
    \ee
    Writing
    $$
    (H_0 - \Lambda_q \mp \lambda)^{-1} = \mp \lambda^{-1} p_q + (I-p_q)(H_0 - \Lambda_q \mp \lambda)^{-1},
    $$
    bearing in mind that the operator $(I-p_q)(H_0 - \Lambda_q \mp \lambda)^{-1}$ admits a uniform limit as $\lambda \downarrow 0$, and applying the Weyl inequalities \eqref{o30}, we easily find that  for each $\veps \in (0,1)$ we have
    $$
    n_+((1+\veps)\lambda; \opwv^{1/2} p_q \opwv^{1/2}) + \cO_{\veps, q}(1) \leq
    $$
    $$
    n_\mp(1; \opwv^{1/2} (H_0 - \Lambda_q \mp \lambda)^{-1} \opwv^{1/2}) \leq
    $$
    \bel{n62}
    n_+((1-\veps)\lambda; \opwv^{1/2} p_q \opwv^{1/2}) + \cO_{\veps, q}(1),
    \ee
    as $\lambda \downarrow 0$. Further, by Corollaries \ref{nf1}, \ref{nf10} and \ref{df1}, we have
    $$
    n_+(s; \opwv^{1/2} p_q \opwv^{1/2}) = n_+(s; p_q \opwv p_q) =
    $$
    \bel{n63}
    n_+(s; \opw(v_{b,q})) = n_+(s;\opaw(\tvbq)) = n_+(s; p_0 \obq p_0) = n_+(s; p_r \zeta p_r), \quad s>0.
    \ee
    Putting together \eqref{n61a}, \eqref{n62}, and \eqref{n63}, we get
    \bel{fin31}
    n_+((1+\veps)\lambda;p_r \, \zeta \, p_r) + \cO_{\veps, q}(1) \leq \cN_q(\lambda; \pm \opwv) \leq  n_+((1-\veps)\lambda;p_r \, \zeta \,p_r) + \cO_{\veps, q}(1).
    \ee
    By \cite{RaWa02}, $n_+(\lambda; p_r \zeta p_r)$ tends to infinity as $\lambda \downarrow 0$ which implies ${\rm rank}\,(p_r \zeta p_r) = \infty$. By \eqref{fin31}, the counting function $\cN_q(\lambda; \pm \opwv)$ also  tends to infinity as $\lambda \downarrow 0$, and hence $m_q^\pm = \infty$. Finally,
    estimate \eqref{n60a} follows easily from \eqref{fin31}.
    \end{proof}
    Let $\Gamma \subset \rd$ be a Jordan curve, i.e. a simple closed curve. We will call it $C^2$-smooth if there exists a $C^2$-smooth diffeomorphism ${\bf x} : {\mathbb S}^1 \to \Gamma$.
    %Then we will call $\Gamma : = {\bf x}({\mathbb S}^1)$ a {\em $C^2$-smooth Jordan curve}.\\
    \begin{proposition} \label{pfin1}
    Let $\Omega \subset \rd$ be a bounded domain. Then there exists a sequence of  $C^2$-smooth Jordan curves $\Gamma_j \subset \Omega$ such that
    \bel{fin51}
    \lim_{j \to \infty} \gC(\Gamma_j) = \gC(\overline{\Omega}).
    \ee
    \end{proposition}
    \begin{proof}
    We follow closely the main idea suggested by \cite{Ran19}. Let $K \subset \rd$ be a compact set, and $j \geq 2$. Set
    $$
    \Delta_j(K) : = \max_{w_1,\ldots,w_j} \prod_{k,\ell, k \neq \ell} |w_k-w_\ell|.
    $$
    Then \cite[Satz 4]{Pom64} implies that if, in addition, $K$ is  connected, we have
    \bel{fin10}
    j^j \gC(K)^{j(j-1)} \leq \Delta_j(K) \leq (4e^{-1} \ln{j} + 4)^j j^j \, \gC(K)^{j(j-1)}.
    \ee
    By the left-hand inequality in \eqref{fin10}, there exist $w_1,\ldots,w_j$ in $\overline{\Omega}$ such that
    \bel{fin40}
    \prod_{k,\ell, k \neq \ell} |w_k-w_\ell| \geq j^j \gC(\overline{\Omega})^{j(j-1)}.
    \ee
    Choosing points $w'_k \in \Omega$ sufficiently close to $w_k$ we can find $w_1',\ldots,w_j' \in \Omega$ such that
     $$
    \prod_{k,\ell, k \neq \ell} |w'_k-w'_\ell| \geq \gC(\overline{\Omega})^{j(j-1)}.
    $$
    Then there exists a $C^2$-smooth Jordan curve
    \bel{fin50}
    \Gamma_j = \left\{{\bf x}(s) \in \Omega \, | \, s \in \sone\right\}
    \ee
    such that $w_k' \in \Gamma_j$, $k=1,\ldots,j$.  In order to see this, it suffices to connect the points $w_k'$, $k=1,\ldots, j$, by a piecewise-linear Jordan curve lying in $\Omega$, and then smooth out the corners. Since
    $\Gamma_j$ is compact and connected, \eqref{fin40} implies
    \bel{fin11}
    \Delta_j(\Gamma_j) \geq \gC(\overline{\Omega})^{j(j-1)}.
    \ee
    On the other hand, the right-hand inequality of \eqref{fin10} implies
    \bel{fin12}
    \Delta_j(\Gamma_j) \leq (4e^{-1} \ln{j} + 4)^j j^j \gC(\Gamma_j)^{j(j-1)}.
    \ee
    Combining \eqref{fin11} and \eqref{fin12}, we get
    $$
    \gC(\overline{\Omega}) \leq (4e^{-1} \ln{j} + 4)^{1/(j-1)} j^{1/(j-1)} \gC(\Gamma_j),
    $$
    and, therefore,
    \bel{fin13}
    \liminf_{j \to \infty} \gC(\Gamma_j) \geq \gC(\overline{\Omega}).
    \ee
    Since $\Gamma_j \subset \overline{\Omega}$, we have
    $\gC(\Gamma_j) \leq \gC(\overline{\Omega})$
    which together with \eqref{fin13} yields \eqref{fin51}.
    \end{proof}
\begin{corollary} \label{corfin1}
    Let $\Omega \subset \rd$ be a bounded domain. Then there exists a sequence of domains $\Omega_j \subset \rd$ with Lipschitz boundaries $\partial \Omega_j$, such that $\overline{\Omega}_j \subset \Omega$ and
    \bel{fin8}
    \lim_{j \to \infty} \gC(\overline{\Omega}_j) = \gC(\overline{\Omega}).
    \ee
    \end{corollary}
    \begin{proof}
    Let $\Gamma_j$ be the $C^2$-smooth Jordan curve introduced in \eqref{fin50}, and let ${\bf n}(s) : = \frac{(x_2'(s), -x_1'(s))}{|{\bf x}'(s)|}$, $s \in \sone$,
    be a normal unit vector to $\Gamma_j$. Set
    $$
    \Omega_j : = \left\{{\bf x}(s) + t {\bf n}(s) \, | \, s \in \sone, \; |t| < \varepsilon_j\right\}
    $$
    where $\varepsilon_j > 0$ is so small that $\overline{\Omega}_j \subset \Omega$ and $\partial \Omega_j$ is Lipschitz-smooth. Since
    $\Gamma_j \subset \overline{\Omega}_j  \subset \overline{\Omega}$,
    \eqref{fin8} follows from \eqref{fin51}.
    \end{proof}
    \begin{proposition} \label{pfin2} \cite[Lemma 2]{FiPu06}
    Let $\Omega \subset \rd$ be a bounded domain with Lipschitz boundary. Fix $q \in \Z_+$. Then ${\rm rank}\,p_q \, \one_{\overline{\Omega}} \, p_q = \infty$ and we have
    \bel{fin9}
   \lim_{k \to \infty} \nu_k(p_q \, \one_{\overline{\Omega}} \, p_q) = -k \ln{k} + \left(1+ \ln{\left(\frac{b \gC(\overline{\Omega})^2}{2}\right)}\right)k + o(k), \quad k \to \infty.
    \ee
    \end{proposition}

    \subsubsection{Proof of Theorem \ref{nth1}} Pick a sequence of domains $\Omega_j \subset \rd$ with Lipschitz boundaries $\partial \Omega_j$ such that $\overline{\Omega}_j \subset \Omega$ and \eqref{fin8} holds true; the existence of such a sequence is guaranteed by Corollary \ref{corfin1}.  Set
    $$
    m_j^- : = \inf_{\bx \in \Omega_j}\,\zeta(\bx), \quad j \in \N, \quad m^+ : = \sup_{\bx \in \Omega}\zeta(\bx).
    $$
    Evidently, $0 < m_j^- \leq m^+ < \infty$. Moreover,
   \bel{n109}
    m_j^- \, \one_{\overline{ \Omega}_j}(\bx) \leq \zeta(\bx) \leq  m^+ \, \one_{\overline{\Omega}}(\bx), \quad \bx \in \rd, \quad j \in \N.
    \ee
    By the mini-max principle, estimate \eqref{n109} implies
   \bel{n108}
    m_j^- \, \nu_k(p_r \, \one_{\overline{\Omega}_j} \, p_r) \leq \nu_k(p_r \zeta p_r) \leq  m^+ \, \nu_k(p_r \, \one_{\overline{\Omega}} \, p_r), \quad k \in \Z_+.
    \ee
    By \eqref{fin9} and \eqref{n108}, we get
    $$
    1 + \ln{\left(\frac{b\gC(\overline{\Omega}_j)^2}{2}\right)} \leq  \liminf_{k \to \infty} \frac {\ln{\nu_k(p_r \, \zeta \, p_r)} + k\ln{k} }{k} \leq
    $$
    \bel{n110}
    \limsup_{k \to \infty} \frac {\ln{\nu_k(p_r \, \zeta \, p_r)} + k\ln{k} }{k} \leq  1 + \ln{\left(\frac{b\gC(\overline{\Omega})^2}{2}\right)},
    \ee
    for every $j$. Combining \eqref{n110} and \eqref{fin8}, we obtain
     \bel{n111}
   \ln{\nu_k(p_r \, \zeta \, p_r)} = - k\ln{k} + \left(1 + \ln{\left(\frac{b\gC(\overline{\Omega})^2}{2}\right)}\right)k  + o(k), \quad k \to \infty.
   \ee
   Now \eqref{1} follows from \eqref{n60a} and \eqref{n111}.
   % \end{proof}

\subsubsection{Proof of Theorem \ref{nth2}}
For  $\delta \in \re$, $c_0>0$, $c_1 \in \re$, and $R>0$,  set
$$
\chi_{\delta, c_0, c_1, R}(\bx) : = c_0 |\bx|^\delta e^{-\gamma|\bx|^{2\beta}}\one_{\rd \setminus B_R}(\bx) + c_1 \one_{B_R}(\bx), \quad \bx \in \rd,
$$
where
$\beta>0$, $\gamma>0$, are the parameters introduced in the statement of the theorem.
Arguing as in the proof of \cite[Theorem 2.2]{LuRa15}, we can show that there exist $\delta_< \leq \delta_> \in \re$, $0 \leq c_{0,<} \leq c_{0,>}$,
$c_{1,<} \leq c_{1,>} \in \re$ and $R>0$, such that
\bel{n77}
    \nu_k(p_0 \,\chi_{\delta_<, c_{0,<}, c_{1,<}, R} \, p_0) \leq \nu_k(p_r \, \zeta \, p_r) \leq  \nu_k(p_0 \, \chi_{\delta_>, c_{0,>}, c_{1,>}, R} \, p_0)), \quad k \in \Z_+.
    \ee
    Since the functions $\chi_{\delta, c_0, c_1, R}$ are radial, we easily check that the eigenvalues of the operator $p_0 \,\chi_{\delta, c_{0}, c_{1}, R} \, p_0 : {\rm Ran}\,p_0 \to {\rm Ran}\,p_0$ coincide with the numbers
    $$
    \langle \chi_{\delta, c_0, c_1, R}\,\varphi_{k,0}, \varphi_{k,0}\rangle_{L^2(\rd)} = \frac{1}{k!} \left((2/b)^{\delta/2} c_0\int_\rho^\infty t^{k+\delta/2} e^{-\mu t^\beta -t} dt + c_1 \int_0^\rho e^{-t} t^{k} dt\right), \quad k \in \Z_+,
    $$
    with
    $\mu =  (2/b)^\beta \gamma$ and $\rho = bR^2/2$.
    Applying \eqref{n77} and \cite[Lemma 5.3]{LuRa15}, we find that
    \bel{o5a}
    \ln{\nu_k(p_r \, \zeta \, p_r)} = \left\{
    \begin{array} {l}
    -\sum_{1 \leq j < \frac{1}{1-\beta}} f_j k^{(\beta-1)j+1} +\cO(\ln{k}) \quad {\rm if} \quad \beta \in (0,1), \\[2mm]
    -\left(\ln{(1+\mu)}\right)\,k + \cO(\ln{k}) \quad {\rm if} \quad \beta = 1, \\[2mm]
    -\frac{\beta-1}{\beta} k \ln{k} + k \frac{\beta-1-\ln{(\mu\beta)}}{\beta} \\[2mm]
    -\sum_{1 \leq j < \frac{\beta}{\beta-1}} g_j k^{(\frac{1}{\beta}-1)j+1} +\cO(\ln{k}) \quad {\rm if} \quad \beta \in (1,\infty),
    \end{array}
    \right.
    \ee
    as $k \to \infty$, the coefficients $f_j$ and $g_j$ being introduced in the statement of Theorem \ref{nth2}.  Now asymptotic relations \eqref{4} -- \eqref{6}  follow from \eqref{n60a}, and \eqref{o5a}.
   % \end{proof}

\subsection{Proof of Theorem \ref{oth1}}
    \label{ss43}

    Our first step, Proposition \ref{op1} below, reduces the asymptotic analysis of $\cN_q^>(\lambda; \opwv)$ and $\cN_q^<(\lambda; \opwv)$ as $\lambda \downarrow 0$, to the eigenvalue asymptotics for the Toeplitz operator $\cT_q(\cV)$, $q \in \Z_+$. In fact, we formulate Proposition \ref{op1} in a more general setting.

    \begin{proposition} \label{op1}
    Let $S = S^* \in \gB(L^2(\rd))$ such that the operator $S H_0^{-1}$ is compact. Let
    $T: = \Lambda_q \left(|{\rm Re}\,(S H_0^{-1})| + |{\rm Im}\,(S H_0^{-1})|\right)$.
    Then for any $q \in \Z_+$ and $\veps > 0$ we have
    \bel{o3}
    n_+(\lambda; p_q(S - \veps |T|)p_q) + \cO_{q,\veps}(1) \leq \cN_q^+(\lambda; S) \leq n_+(\lambda; p_q(S + \veps |T|)p_q) + \cO_{q,\veps}(1),
    \ee
    \bel{o4}
    n_-(\lambda; p_q(S + \veps |T|)p_q) + \cO_{q,\veps}(1) \leq \cN_q^-(\lambda; S) \leq n_-(\lambda; p_q(S - \veps |T|)p_q) + \cO_{q,\veps}(1),
    \ee
    as $\lambda \downarrow 0$.
    \end{proposition}
    We omit the standard proof which follows the general lines of \cite[Section 5]{Ra90}. \\
    Now note that under the hypotheses of Proposition \ref{op1}, the Weyl inequalities \eqref{o30} and the mini-max principle easily imply
    \bel{n500}
     n_\pm(\lambda; p_q(S \mp \veps |T|)p_q) \geq n_\pm(\lambda(1+\eta); p_q S p_q) - 2 n_+(\lambda^2 \eta^2 \veps^{-2}; (1 + \Lambda^{-1}_0 \Lambda_q)^2 p_q S^2 p_q),
     \ee
     \bel{n501}
     n_\pm(\lambda; p_q(S \pm \veps |T|)p_q) \leq n_\pm(\lambda(1-\eta); p_q S p_q) + 2 n_+(\lambda^2 \eta^2 \veps^{-2}; (1 + \Lambda^{-1}_0 \Lambda_q)^2 p_q S^2 p_q),
     \ee
     for any $\lambda>0$, $\veps>0$, and $\eta \in (0,1)$. Combining \eqref{o3} - \eqref{o4} and \eqref{n500} - \eqref{n501} with $S = \opwv$, and bearing in mind Corollary \ref{nf1}, we obtain

     \becol{nf600}
     Under the hypotheses of Theorem \ref{oth1} there exists a constant $C_1>0$ such that
      $$
     n_-(\lambda(1+\eta); \opw(\vbq)) - 2 n_+(C_1 \lambda^2 \eta^2 \veps^{-2}; \opw(\wbq)) + \cO_{\veps, q}(1) \leq
     $$
     $$
     \cN_q^<(\lambda) \leq
     $$
     \bel{n502}
      n_-(\lambda(1-\eta); \opw(\vbq)) + 2 n_+(C_1 \lambda^2 \eta^2 \veps^{-2}; \opw(\wbq)) + \cO_{\veps, q}(1),
      \ee
     $$
     n_+(\lambda(1+\eta); \opw(\vbq)) - 2 n_+(C_1 \lambda^2 \eta^2 \veps^{-2}; \opw(\wbq)) + \cO_{\veps, q}(1) \leq
     $$
     $$
     \cN_q^>\lambda) \leq
     $$
     \bel{n503}
      n_+(\lambda(1-\eta); \opw(\vbq)) + 2 n_+(C_1 \lambda^2 \eta^2 \veps^{-2}; \opw(\wbq)) + \cO_{\veps, q}(1),
      \ee
      \eco

    Our next goal is to study the asymptotics  of   $n_\pm(\lambda; \opw(\vbq))$ and $n_+(\lambda; \opw(\wbq))$ as $\lambda \downarrow 0$. To this end, we will apply the approach developed in \cite{DaRo87}.

    \begin{proposition} \label{op2}
    Let $\cF = \overline{\cF} \in \cS^{-\gamma}_\rho(\re^{2n})$ for some $n \in \N$, $\gamma > 0$, and $\rho \in (0,1]$. Assume that the functions $\gV_n^+(\cdot;\cF)$  and $\gV_n^-(\cdot;\cF)$  satisfy the condition $\cC$, and
     $$\liminf_{\lambda \downarrow 0} \lambda^{2n/\gamma} \gV_n^\pm(\lambda; \cF) > 0.$$  Then there exists $\delta > 0$ such that
    $$
    n_\pm(\lambda; \opw(\cF)) = \gV_n^\pm(\lambda; \cF) (1 + \cO(\lambda^\delta)), \quad \lambda \downarrow 0.
    $$
    \end{proposition}
    Proposition \ref{op2} follows from the main theorem of \cite{DaRo87} with $\varphi(x,\xi) = \phi(x,\xi) : = (1+|x|^2 + |\xi|^2)^{\varrho/2}$ and
    $m(x,\xi) : = (1+|x|^2 + |\xi|^2)^{-\gamma/2}$, $(x,\xi) \in \re^{2n}$.\\ Let us now prove the validity of \eqref{o11}.
     By Proposition \ref{op2},
    \bel{o23}
    n_+(s; \opw(v_{b,q})) = \gV_1^+(s; v_{b,q}) (1 + o(1)), \quad s \downarrow 0.
    \ee
    Since $\gV_1(\cdot;v_{b,q})$ satisfies by assumption condition $\cC$, we find that
    \bel{o24}
    (1+\eta)^{-\gamma_1} \gV_1^+(\lambda; v_{b,q}) \leq \gV_1^+((1+\eta)\lambda; v_{b,q}), \quad
    \gV_1^+((1-\eta)\lambda; v_{b,q}) \leq (1-\eta)^{-\gamma_2} \gV_1^+(\lambda; v_{b,q}),
    \ee
    for any $\eta \in (0,1)$ and $\lambda > 0$. It is easy to check that our assumption $w_{b,q} \in \cS_\varrho^{-2\gamma}(\rd)$ implies the existence of a constant $C_2$ such that
    \bel{o17}
    n_+(s; \opw(w_{b,q})) \leq C_2 s^{-1/\gamma}
    \ee
    for $s>0$ small enough.
    Putting together \eqref{n502} - \eqref{n503} and \eqref{o23} -- \eqref{o17}, we find that there exists a constant $C_3$ such that for any $\eta \in (0,1)$ and $\veps > 0$ we have
    $$
    (1+\eta)^{-\gamma_2} - C_3 (\eta^2 \veps^{-2})^{-1/\gamma} \leq \liminf_{\lambda \downarrow 0} \frac{\cN_q(\lambda; \opwv)}{\gV_1^+(\lambda; v_{b,q})} \leq
    $$
    $$
    \limsup_{\lambda \downarrow 0} \frac{\cN_q(\lambda; \opwv)}{\gV_1^+(\lambda; v_{b,q})} \leq (1-\eta)^{-\gamma_1} + C_3 (\eta^2 \veps^{-2})^{-1/\gamma}.
    $$
     Choosing $\eta = \sqrt{\veps}$ and letting $\veps \downarrow 0$, we obtain \eqref{o11}.
    The proof of \eqref{o12} is  analogous.

    \pagebreak
    {\bf Acknowledgements}. The authors are very grateful to Thomas Ransford who gave in \cite{Ran19} the idea of the proof of Proposition \ref{pfin1}. Moreover, they thank Luchezar Stoyanov for a useful discussion of the details of this proof, as well as Dimiter Balabanski and Hajo Leschke for valuable comments on the applications of non-local potentials in nuclear physics. The partial
support of the Chilean Science Foundation {\em Fondecyt} under Grant 1170816 is gratefully acknowledged.

{\sc Esteban C\'ardenas, Georgi Raikov, Ignacio Tejeda}\\
Facultad de Matem\'aticas\\
Pontificia Universidad Cat\'olica de Chile\\
Av. Vicu\~na Mackenna 4860, Santiago de Chile\\
E-mails: encardenas@uc.cl, graikov@mat.uc.cl, ijtejeda@uc.cl\\


\begin{thebibliography} {[10]}
\frenchspacing \baselineskip=12 pt plus 1pt minus 1pt
\bibitem{AlAvDeHe94} {\sc S. Alama, M. Avellaneda,  P. A. Deift, R. Hempel}, {\em On the existence of eigenvalues of a divergence-form operator $A+\lambda B$ in a gap of $\sigma(A)$},
 Asymptotic Anal. {\bf 8} (1994), 311�-344.

\bibitem{BeCo77} {\sc A. M. Berthier, P. Collet}, {\em Existence and completeness of the wave operators in scattering theory with momentum-dependent potentials}, J. Funct. Anal. {\bf 26} (1977), 1--15.

\bibitem{Bi91}{\sc M. Sh. Birman}, {\em Discrete spectrum in the gaps of a continuous one for perturbations with large coupling constant}, In: Estimates and asymptotics for discrete spectra of integral and differential equations, 57--73, Adv. Soviet Math., {\bf 7}, Amer. Math. Soc., Providence, RI, 1991.

\bibitem{BiSo87}{\sc  M.\u{S}.Birman, M.Z.Solomjak},
{\em Spectral Theory of Self-Adjoint Operators  in  Hilbert
Space}, D. Reidel Publishing Company, Dordrecht, 1987.



\bibitem{Bo99} {\sc A. Boulkhemair}, {\em $L\sp 2$ estimates for
Weyl quantization}. J. Funct. Anal. {\bf 165} (1999),  173--204.

\bibitem{BrPuRa04} {\sc V. Bruneau, A. Pushnitski, G. D. Raikov},
{\em Spectral shift function in strong magnetic fields}, Algebra i
Analiz {\bf 16} (2004), 207--238; see also St. Petersburg Math. J.
{\bf 16} (2005), 181--209.

\bibitem {CaVa72}{\sc A.-P. Calder\'on, R. Vaillancourt}, {\em A class of bounded pseudo-differential operators},
 Proc. Nat. Acad. Sci. U.S.A. {\bf 69} (1972), 1185--1187.

 \bibitem{ChMoKiChKiSo14} {\sc H. G. Cho, B. Y. Moon, K. S. Kim, M.-K. Cheoun, T. H. Kim, W. Y. So}, {\em Comparing a local poetntial with non-local potential for the bound states of $^{16}{\rm O}$ and $^{40}{\rm Ca}$}, Chinese J. Phys., {\bf 52} (2014), 729--737.

 \bibitem {Co75}{\sc H. O. Cordes}, {\em On compactness of commutators of multiplications and convolutions,
 and boundedness of pseudodifferential operators}, J. Funct. Anal. {\bf 18} (1975), 115--131.

  \bibitem {CoArMa70}{\sc M. Coz, L. G. Arnold, A. D. MacKellar}, {\em Nonlocal potentials and their local equivalents}, Ann. Physics {\bf 59} (1970), 219--247.

 \bibitem{DaRo87}  {\sc M. Dauge, D. Robert}, {\em Weyl's formula for a class of pseudodifferential operators with negative order on $L^2({\bf R}^n)$},  Pseudodifferential operators (Oberwolfach, 1986), 91--122, Lecture Notes in Math., {\bf 1256}, Springer, Berlin, 1987.

 \bibitem {DiSj99}{\sc M.Dimassi, J.Sj\"ostrand},
{\em Spectral Asymptotics in the Semi-Classical Limit}. London
Mathematical Society Lecture Notice Series {\bf 268}. Cambridge: Cambridge University Press. 1999.

 \bibitem {Du90}{\sc A. J. Dur\'an}, {\em Laguerre expansions of tempered distributions and generalized functions}, J. Math. Anal. Appl. {\bf 150} (1990),  166--180.



          % \bibitem{FeRa04}
%{\sc C.~Fern{\'a}ndez, G.~Raikov}, {\em  On the singularities of the magnetic
 % spectral shift function at the {L}andau levels}, Ann. Henri Poincar\'e
 % \textbf{5} (2004), 381--403.

      \bibitem {Fe58}{\sc H. Feshbach}, {\em Unified theory of nuclear reactions}, Ann. Physics {\bf 5} (1958), 357--390.


     \bibitem {FeWa71}{\sc A. L. Fetter, J. D. Walecka}, {\em Quantum Theory of Many-Particle Systems}, McGraw-Hill, New York, 1971.

\bibitem{FiPa90}{\sc  A. L. Figotin, L. A. Pastur}, {\em Schr\"odinger operator with a nonlocal potential whose absolutely continuous and point spectra coexist}, Comm. Math. Phys. {\bf 130} (1990),  357--380.

\bibitem {FiPu06} {\sc N. Filonov, A. Pushnitski},
{\em Spectral asymptotics of Pauli operators and orthogonal polynomials in complex domains},
Comm. Math. Phys. {\bf 264} (2006), 759--772.


\bibitem{Fo28}{\sc V. Fock},
   {\em Bemerkung zur Quantelung des harmonischen Oszillators im
Magnetfeld},  Z. Physik {\bf 47} (1928), 446--448.

 \bibitem {GoKaPe16}{\sc M. Goffeng, A. Kachmar, M. Persson Sundqvist}, {\em Clusters of eigenvalues for the magnetic Laplacian with Robin condition}, J. Math. Phys. {\bf 57} (2016), no. 6, 063510, 19 pp.

 \bibitem{GrRy65}{\sc I. S.Gradshteyn, I. M.Ryzhik},
{\em Table of Integrals, Series, and Products}, Academic Press,
New York, 1965.

\bibitem{Ha00} {\sc B. C. Hall},
 {\em Holomorphic methods in analysis and mathematical physics},
 In: First Summer School in Analysis and Mathematical Physics, Cuernavaca
 Morelos, 1998, 1--59, Contemp.Math. {\bf 260}, AMS, Providence, RI, 2000.

\bibitem{HoIII94} {\sc L. H\"ormander}, {\em The Analysis of Linear Partial
Differential Operators. III. Pseudo-Differential Operators},
Corrected Second Printing, Springer-Verlag, Berlin-Heidelberg-New
York-Tokyo, 1994.

\bibitem{Iv98}  {\sc V. Ivrii},
{\em Microlocal Analysis and Precise Spectral Asymptotics},
Springer Monographs in Mathematics. Springer-Verlag, Berlin, 1998.

 \bibitem{JaPiPr17}{\sc  S. Jak\v{s}i\'{c}, S. Pilipovi\'{c}, B. Prangoski}, {\em $G$-type spaces of ultradistributions over $\re_+^d$ and the Weyl pseudo-differential operators with radial symbols}, Rev. R. Acad. Cienc. Exactas F\'is. Nat. Ser. A Math. RACSAM {\bf 111} (2017), 613--640.

   \bibitem{Je77}{\sc A. Jensen}, {\em Some remarks on eigenfunction expansions for Schr\"odinger operators with non-local potentials}, Math. Scand. {\bf 41} (1977), 347--357.

     \bibitem{KlRa09}{\sc F. Klopp, G. Raikov}, {\em The fate of the Landau levels under perturbations of constant
sign}, Int. Math. Res. Notices, {\bf 2009} (2009), 4726-4734.

\bibitem{La30} {\sc L. Landau}, {\em Diamagnetismus der Metalle}, Z. Physik
 {\bf 64} (1930), 629--637.

% \bibitem{La72} {\sc
% N. S. Landkof},
% {\em Foundations of Modern Potential Theory},
% Die Grundlehren der mathematischen Wissenschaften, {\bf 180}, Springer-Verlag, New York-Heidelberg, 1972.

\bibitem{LuRa15} {\sc  T. Lungenstrass, G. Raikov}, {\em Local spectral asymptotics for metric perturbations of the Landau Hamiltonian}, Anal. PDE {\bf 8} (2015), 1237--1262.

     \bibitem{Pe09} {\sc M. Persson}, {\em Eigenvalue asymptotics of the even-dimensional exterior Landau-Neumann Hamiltonian}, Adv. Math. Phys. {\bf 2009} (2009), Art. ID 873704, 15 pp.

         \bibitem{Pom64} {\sc Ch. Pommerenke},  {\em \"Uber die Faberschen Polynome schlichter Funktionen}, Math. Z. \textbf{85} (1964), 197--208.


\bibitem{PuRaVBl13} {\sc A. Pushnitski, G. Raikov, C.
Villegas-Blas},  {\em Asymptotic density of eigenvalue clusters for
the perturbed Landau Hamiltonian}, Commun. Math.
Phys. {\bf 320} (2013), 425--453.

 \bibitem{PuRo07} {\sc A. Pushnitski, G.  Rozenblum} , {\em Eigenvalue clusters of the Landau Hamiltonian in the exterior of a compact domain}, Doc. Math. {\bf 12} (2007), 569--586.

\bibitem{Ra90}
{\sc G.~Raikov}, {\it  Eigenvalue asymptotics for the {S}chr\"odinger operator with
  homogeneous magnetic potential and decreasing electric potential. {I}.
  {B}ehaviour near the essential spectrum tips}, Comm. Partial Differential
  Equations {\bf 15} (1990), 407--434.



  \bibitem{RaWa02}
{\sc G.~Raikov, S.~Warzel}, {\it  Quasi-classical versus non-classical spectral
  asymptotics for magnetic {S}chr\"odinger operators with decreasing electric
  potentials}, Rev. Math. Phys. {\bf 14} (2002), 1051--1072.

  \bibitem{Ran95}{\sc T. Ransford}, {\em Potential Theory in the Complex Plane}, Cambridge University Press, 1995.

\bibitem{Ran19} {\sc T. Ransford}, {\em On the logarithmic capacity of the closure of a domain}, private communication, (2019)

  \bibitem{RoTa09}
{\sc  G. Rozenblum, G. Tashchiyan}, {\em On the spectral properties of the perturbed Landau Hamiltonian},  Comm. Partial Differential Equations {\bf 33} (2008),  1048--1081.

  \bibitem{Sa83}
{\sc G. R. Satchler}, {\em Direct Nuclear Reactions}, Clarendon, Oxford, 1983.

\bibitem{Sh01} {\sc M.A.Shubin},
{\em Pseudodifferential Operators and Spectral Theory}, Second
Edition, Berlin etc.: Springer-Verlag (2001).

\bibitem{SoKi97} {\sc W.-Y. So, B.-T. Kim}, {\em Calculations of bound states for nonlocal potentials}, J. Korean Phys Society {\bf 30} (1997), 175--179.

\bibitem{Te19} {\sc I. Tejeda}, {\em Spectral Asymptotics for Compactly Supported Electric Perturbations of the Landau Hamiltonian}, M.Sc. Thesis, Facultad de Matemáticas, Pontificia Facultad Cat\'olica de Chile, 2019.

 \bibitem{Un16} {\sc A. Unterberger}, {\em Which pseudodifferential operators with radial symbols are non-negative?}, J. Pseudo-Differ. Oper. Appl. {\bf 7} (2016), 67--90.\\


\end{thebibliography}
\end{document}